\documentclass[a4paper,onecolumn,11pt,accepted=2024-01-01]{quantumarticle}
\pdfoutput=1
\usepackage[numbers,sort&compress]{natbib}
\usepackage{hyperref,url}
\usepackage{amsmath}
\usepackage{mathtools}
\usepackage{amsfonts}
\hypersetup{breaklinks,colorlinks=true,linkcolor=blue,citecolor=blue,urlcolor=blue,filecolor=blue}
\usepackage[dvipsnames, x11names]{xcolor}
\usepackage{graphicx}
\usepackage{dcolumn}

\usepackage[ 
margin=.85in,
]{geometry}
\usepackage{booktabs}
\usepackage{physics}
\usepackage{amsthm} 
\usepackage{thmtools}
\usepackage{thm-restate}
\theoremstyle{definition}
\usepackage{cleveref} 

\usepackage{tikz}
\usetikzlibrary{quantikz}
\usepackage{caption}
\usepackage{subcaption}
\usepackage[english]{babel}

\makeatletter
\def\bbl@set@language#1{%
  \edef\languagename{%
    \ifnum\escapechar=\expandafter`\string#1\@empty
    \else\string#1\@empty\fi}%
  \@ifundefined{babel@language@alias@\languagename}{}{%
    \edef\languagename{\@nameuse{babel@language@alias@\languagename}}%
  }%
  \select@language{\languagename}%
  \expandafter\ifx\csname date\languagename\endcsname\relax\else
    \if@filesw
      \protected@write\@auxout{}{\string\select@language{\languagename}}%
      \bbl@for\bbl@tempa\BabelContentsFiles{%
        \addtocontents{\bbl@tempa}{\xstring\select@language{\languagename}}}%
      \bbl@usehooks{write}{}%
    \fi
  \fi}
\newcommand{\DeclareLanguageAlias}[2]{%
  \global\@namedef{babel@language@alias@#1}{#2}%
}
\makeatother

\DeclareLanguageAlias{en}{english}

\newcommand{\citen}[1]{Ref.~\citenum{#1}}

\newtheorem{definition}{Definition}[section]

\usepackage[colorinlistoftodos, textsize=small, color=orange!50]{todonotes}

\newcommand{\idmat}[0]{\mathbb{I}}

\newcommand{\bigo}[1]{\mathcal{O}(#1)}
\newcommand{\bigot}[1]{\tilde{\mathcal{O}}( #1)}
\newcommand{\bigt}[1]{\Theta(#1)}
\newcommand{\bigtt}[1]{\tilde{\Theta}(#1)}
\newcommand{\bigomega}[1]{\Omega(#1)}
\newcommand{\defeq}[0]{: =}

\newcommand{\Google}{\affiliation{Google Quantum AI, Venice, CA, USA}}

\begin{document}

\title{Accelerating Quantum Algorithms with\newline Precomputation}

\author{William J. Huggins}
\email{whuggins@google.com}
\Google

\author{Jarrod R. McClean}
\Google

\date{2024-02-15}

\begin{abstract}
    Real-world applications of computing can be extremely time-sensitive.
    It would be valuable if we could accelerate such tasks by performing some of the work ahead of time.
    Motivated by this, we propose a cost model for quantum algorithms that allows quantum precomputation; i.e., for a polynomial amount of ``free'' computation before the input to an algorithm is fully specified, and methods for taking advantage of it.
    We analyze two families of unitaries that are asymptotically more efficient to implement in this cost model than in the standard one.
    The first example of quantum precomputation, based on density matrix exponentiation, could offer an exponential advantage under certain conditions.
    The second example uses a variant of gate teleportation to achieve a quadratic advantage when compared with implementing the unitaries directly.
    These examples hint that quantum precomputation may offer a new arena in which to seek quantum advantage.
\end{abstract}

\maketitle

\section{Introduction}

In order to efficiently use limited computational resources, it is natural to quantify and minimize their use.
In quantum computing, we frequently try to minimize some proxy for the spacetime cost of an algorithm, such as the number of two-qubit gates on an near-term machine or the number of non-Clifford gates on a fault-tolerant device.
Focusing on spacetime metrics allows one to easily incorporate the fungibility of additional qubits and time inside error correcting codes~\cite{Fowler2012-ti,Gidney2019-qi,Litinski2019-nu}, as well as elements of algorithmic parallelism.
However, in some cases, one is interested in the raw time to solution, or ``wall-clock time,'' given any reasonable resources.
As such, in this paper, we explore a different cost model that allows for what we call ``quantum precomputation.''
In the process, we aim to understand the opportunities and challenges inherent in generalizing classical ideas of precomputation, e.g., caching of results, indexing in databases, or creating lookup tables.
The precomputation cost model allows for a quantum algorithm to start with access to a specially prepared resource state that depends on the algorithm and some portion (but not all) of its input.
We neglect the cost of preparing this resource state, but we demand that it can be prepared efficiently, i.e., that the quantum and classical resources required scale polynomially in the size of the input.

Our precomputation cost model is motivated by real-world problems where the crucial limited resource is the computational power available after the problem is fully specified.
For some of these problems, the value of finding a solution as quickly as possible would justify investing extra effort ahead of time preparing to perform a computation.
In fields ranging from optimization, to finance, to data analysis, there are tasks that naturally fit into this framework.
If we can build useful quantum primitives that accelerate such tasks in the precomputation cost model, it could have a substantial impact even in cases where the overall quantum advantage is modest or non-existent.
We study quantum precomputation because of these potential practical applications, and also because it offers the chance to investigate the nature of quantum computation from another angle.
Notably, the no-cloning theorem imposes limitations on our ability to reuse the results of earlier quantum computations, which implies that precomputation may occupy a different role in quantum computing than it does classically.

In order for the precomputation cost model to make sense, there must be some components of the computational task that are naturally specified before others.
For example, we could be given a classical description of a Hamiltonian now with the understanding that we will want to estimate some properties of its ground state that will be determined at a later time.
In such a situation, we could prepare for when these properties are specified by generating and storing a sufficient number of copies of the ground state.
In other cases, we might have a classical description of some unitary \(U\) available now that we will later wish to apply to a (currently unknown) state \(\ket{\psi}\).
In this paper, we ask if we can find interesting or useful families of tasks that can be implemented using asymptotically fewer quantum resources in a cost model that allows for free precomputation.

We formalize our definition of the precomputation cost model in \Cref{sec:precomputation}.
In \Cref{sec:prior_work}, we discuss some of the connections that quantum precomputation has with prior work on quantum and classical computation.
We go on to explore how existing algorithmic primitives can interpreted as tools for quantum precomputation in \Cref{sec:precomputation_examples}.
Specifically, we make use of density matrix exponentiation and gate teleportation to accelerate the application of certain unitaries in the precomputation cost model~\cite{Lloyd2014-td,Gottesman1999-gr}, finding the possibility of speedups that range from quadratic to exponential (when comparing the cost in the precomputation model with the usual quantum gate complexity).
In \Cref{sec:selective_gate_teleportation_precomputation}, we present a less straightforward protocol for quantum precomputation that uses a technique known as selective teleportation~\cite{Fowler2012-ti} to yield a quadratic improvement in complexity for a family of diagonal unitaries.
We conclude with a discussion of open questions and potential applications in \Cref{sec:discussion}.

\section{The Precomputation Cost Model}
\label{sec:precomputation}

\subsection{Formalizing the cost model}

Analyzing the resources required to execute an algorithm requires a cost model.
A good cost model encodes useful assumptions that simplify the analysis, abstracting away irrelevant details while keeping the essential information required to answer the questions at hand.
There are a number of different choices one could make in formalizing the intuition behind quantum precomputation into a cost model; i.e., specifying what it means to ``allow a reasonable amount of work to be performed for free.''
In this section, we propose a concrete definition flexible enough to encompass several interesting examples rather than a maximally general abstract definition.

There are many kinds of computational tasks that we might wish to analyze in the precomputation cost model.
We will loosely formalize a computational task as an algorithm, which we treat as a map that takes an input from some set of valid inputs and returns a correct output (or a sample from a correct distribution over outputs).
Different algorithms may define different notions of valid inputs and correct outputs. 
For now, we leave these details unspecified, although they may be crucial to determining the complexity of implementing an algorithm.
For example, there are some tomographic tasks that are efficient for pure state inputs but prohibitively expensive for general mixed state inputs~\cite{Gilyen2022-gu}.
In other cases, the computational complexity of a problem may vary depending on the definition of the ``correct'' output, e.g., what kind of approximation is allowed~\cite{Gharibian2022-qy}.

To be sufficiently general, we need a notion of a quantum algorithm that can accept both quantum and classical input and can output both quantum and classical data.\footnote{We could consider all of the inputs and outputs to be quantum states, but treating them separately will help us take a more nuanced approach that differentiates the quantum and classical resources.}
We also need to allow for the possibility that the input is partitioned into two components that are provided at different times.
For simplicity, we assume that the earlier input (that might be used in the precomputation step) is classical, and that the later input may be a combination of classical and quantum data.
Let \(x\) denote the (classical) input provided at the earlier time and let \(\rho\) and \(y\) denote the quantum and classical components of the input provided at a later time.
For the quantum and classical outputs we use the symbols \(\sigma\) and \(z\) respectively.

In the usual situation, where we do not take advantage of the fact that some portion of the input may be available ahead of time, a quantum algorithm \(\mathcal{A}\) implements a map 
\begin{equation}
    \label{eq:standard_computation_map}
    \mathcal{A} : x, y, \rho \rightarrow  z, \sigma.
\end{equation}
In general, we can understand \(\mathcal{A}\) as performing some classical computation that takes \(x\) and \(y\) as an input, determining a quantum circuit that is subsequently applied to \(\rho\).
The portions of the resulting quantum state that are not measured or discarded constitute \(\sigma\).
The classical component of the output, \(z\), is classically computed from \(x, y\), and the measurement outcomes.
In a standard cost model, we are concerned with the cost of executing the algorithm \(\mathcal{A}\) given access to \(x\), \(y\), and \(\rho\).

In a model that allows for free precomputation, we aim to produce the same (distribution over) outputs by implementing the map 
\begin{equation}
    \label{eq:precomputation_map}
\mathcal{P}: \bar{x}(\mathcal{A}, x), \ket{\Gamma(\mathcal{A}, x)}, y, \rho \rightarrow z, \sigma,
\end{equation}
where \(\bar{x}(\mathcal{A}, x)\) and \(\ket{\Gamma(\mathcal{A}, x)}\) represent the classical and quantum outputs of some precomputation step.
We allow for \(\bar{x}(\mathcal{A}, x)\) and \(\ket{\Gamma(\mathcal{A}, x)}\) to be generated using a ``reasonable'' amount of classical and quantum computation performed ahead of time, i.e., with knowledge of \(\mathcal{A}\) and \(x\) but not \(\rho\) or \(y\).
In a precomputation cost model, the only cost that we consider directly is the cost of performing the map \(\mathcal{P}: \bar{x}(\mathcal{A}, x), \ket{\Gamma(\mathcal{A}, x)}, y, \rho \rightarrow z, \sigma\).
In order to fully define a precomputation cost model and compare it to a standard cost model, we therefore have to specify answers to two questions: 
i) How will we quantify the costs of implementing \(\mathcal{A}\) and \(\mathcal{P}\)?
ii) What do we mean when we say that we allow for a ``reasonable'' amount of classical and quantum computation to be used in the preparation of \(\ket{\Gamma(\mathcal{A}, x)}\) and \(\bar{x}(\mathcal{A}, x)\)?

In this paper, we focus on quantifying the quantum resources used to implement \(\mathcal{P}\) (and \(\mathcal{A}\) itself) in terms of the quantum circuit complexity (a term that we use interchangeably with ``gate complexity''), the number of gates from some elementary set of discrete operations required to implement the algorithm.
We consider a discrete set of gates that consists of one- and two-qubit Clifford gates, single-qubit computational basis measurement operations, and \(T\) gates. 
We also choose to count single-qubit identity operations as gates in order to include the cost of storage (which is comparable in most architectures to the cost of active workspace).
This choice implies that our notion of circuit complexity grows asymptotically as fast as the product of the number of qubits and the circuit depth (the number of layers of gates, executed in parallel).

We could define other related models that allow for free precomputation but account for ``cost'' differently.
Depending on the context, it might be useful to work in an oracle model, or to count only the number of non-Clifford gates, or even to quantify the space-time volume used in a particular error-correcting code.
It might also be useful to discuss the number of gates required for the best known implementation of an algorithm, rather than the absolute minimum required.
For the examples we consider, this distinction will not be important.
We find that discussing the gate complexity is convenient because it allows us to use the same model to consider several different examples, but we will make some comments along the way regarding other notions of cost.
As we consider these examples, it will sometimes make sense to allow for \(\mathcal{A}\) or \(\mathcal{P}\) to be implemented with some error.
In the context of this work, when we need to allow for some notion of error, it will be sufficient to focus on the case where the output is a quantum state and we can quantify the error using a single parameter \(\epsilon\) that bounds the trace distance between the ideal output and the actual output.

By focusing on quantifying the cost in the precomputation model in terms of the number of quantum operations, we are implicitly treating quantum operations as a fundamentally different and more limited resource than classical ones.
This decision is motivated by the practical observation that quantum operations on a fault-tolerant computer are expected to be vastly slower and more expensive than classical operations~\cite{Babbush2021-aq}.
Nevertheless, we would like a definition of the precomputation cost model that is useful in practice, so we demand that the classical time and space complexity of implementing \(\mathcal{P}\) scales as \(\mathcal{O}(\text{poly}(\epsilon^{-1}, |x|, |y|, |\rho|))\).
Here the notation \(|*|\) indicates the size of \(*\) in terms of classical or quantum bits.

Besides specifying how we quantify the cost of implementing \(\mathcal{A}\) or \(\mathcal{P}\), we also need to formalize the notion that the amount of work performed ahead of time is required to be ``reasonable.''
We should bound the quantum gate complexity of the precomputation step, as well the classical time and space complexities.
For all of these resources, we allow their usage during the precomputation step to scale as \(\bigo{\text{poly}(\epsilon^{-1}, |x|)}\).
Although we define our model with this coarse-grained notion of what is allowed during the precomputation step, we will discuss the actual scaling of the various resources in more detail for the particular examples we consider in this paper.

\section{Prior work}
\label{sec:prior_work}

While the authors are not aware of prior work that has focused on a cost model that allows for free precomputation in the sense that we consider, there are a number of closely related ideas that we draw inspiration from.
The paper that first described gate teleportation speculated that it might be used to mass manufacture resource states for later consumption~\cite{Gottesman1999-gr}.
For example, one could imagine using magic state distillation to distill a large number of magic states, storing them for use in a later computation~\cite{Bravyi2005-vi}.
Going beyond the prototypical use of magic state distillation to implement a \(T\) gate, state distillation schemes have been proposed for a variety of other few-qubit operations~\cite{Cody-Jones2012-tm, Jones2013-cf, Campbell2016-ls,Litinski2019-nu,Gidney2019-qi}.
In \citen{Cody-Jones2012-tm}, Jones et al. proposed a method that implements an arbitrary single-qubit \(Z\) rotation with success probability \(1-\delta\) by precomputing and storing a resource state on \(\bigo{-\log(\delta)}\) qubits.
More abstractly, measurement based quantum computing has some similarity to quantum precomputation, but it aims to prepare generically useful resource states rather than ones that are tailored to accelerating particular algorithms~\cite{Raussendorf2001-na, Nielsen2004-oq}.

The idea of precomputing and storing a reservoir of resource states for single or few-qubit operations is appealing, but it faces serious challenges.
In particular, the number of such resource states required for interesting and classically intractable applications appears large~\cite{Berry2019-qo,Sanders2020-lf,Gidney2021-ru}, while quantum memory has a comparable cost with active workspace in most proposed architectures~\cite{Terhal2015-dg}.
For example, \citen{Gidney2021-ru} estimates that thousands of logical qubits and billions of Toffoli and \(T\) gates would be required to factor a 2048 bit RSA integer using Shor's algorithm.
A fault-tolerant quantum computer that large enough to perform this computation, but not too much larger, would be unable to precompute and store more than a tiny fraction of the necessary resource states ahead of time.

Even so, one might ask if precomputing resource states for \(T\) or Toffoli gates offers a simple example of asymptotic advantage when the cost of the precomputation itself is neglected.
In our definition of the precomputation cost model, the answer is no.
This is because, even with access to the appropriate resource state, applying either of these gates still requires a (nonzero) constant number of operations and our model allows \(T\) gates to be performed at unit cost.
If we instead consider the task of implementing arbitrary single-qubit rotations to within some precision \(\epsilon\), \citen{Cody-Jones2012-tm} provides an example where allowing for free precomputation does indeed change the asymptotic cost. Specifically, precomputation can be used to remove the dependence on \(\epsilon\) from the cost (not including the cost of the precomputation step) at the expense of incurring some logarithmic dependence on the allowed failure probability \(\delta\).

The idea of supplementing a quantum computer with a specially-prepared resource state has also been considered from a complexity-theoretic perspective.
The complexity class BQP/qpoly formalizes the power of a polynomial-time quantum computer augmented with an arbitrary resource state, referred to as ``quantum advice,'' that is allowed to depend on the length of the input.
Comparing this complexity class to our model of quantum precomputation requires some care, so we provide a longer discussion in \Cref{app:quantum_advice} and merely summarize the conclusions here.
First of all, the model formalized in BQP/qpoly places no restrictions on the computational power used to prepare the resource state, whereas we require that it be preparable in polynomial time.
Secondly, the quantum advice states of BQP/qpoly can only depend on the length of the input.
We allow for the resource states to depend on a subset of the parameters, denoted by \(x\).
Thirdly, the only problems that fit into the framework of BQP/qpoly are decision problems, which have a classical input and a (single bit of) classical output.
This is a more limited setting than the one that we consider.\footnote{One could imagine analogues of BQP/qpoly that use a similar notion of advice but consider problems beyond the setting of decision problems. The main benefit of focusing exclusively on decision problems is that they are simple to formalize precisely.}

Despite these differences between BQP/qpoly and our notion of quantum precomputation, we can make a useful comparison if we restrict ourselves to considering the power of both models to solve decision problems.
One might suspect that our model of quantum precomputation gets some additional power from the fact that we allow the resource state to depend on the input in richer ways than allowed by the complexity class BQP/qpoly. 
However, any decision problem that is solvable in polynomial time in the precomputation model we have defined is not only a member of BQP/qpoly, but also BQP itself.
This is because we only allow a polynomial amount of ``free'' precomputation, which can't add any power to a machine that is already allowed to run arbitrary polynomial-time quantum computations.
Ultimately, our model of quantum precomputation is trying to capture a finer-grained notion of speedup than these particular complexity classes are designed to address.
Imprecisely, we could say that we are interested in the power of the ``advice that a polynomial time quantum computer can give itself.''

In the context of classical computing, the term ``precomputation'' has been used extensively to describe variations on the idea of performing useful work ahead of time and caching the result.
For example, branch-prediction is an essential component of modern computer architecture design~\cite{Smith1998-gk}.
Precomputation is used to optimize certain tasks in computer graphics~\cite{Sloan2002-zw} and computer vision~\cite{Grady2008-kf}.
The precomputation of expensive operations involved in breaking cryptographic schemes is both a practical and theoretical concern~\cite{Bernstein2013-mn}, which is closely related to the study of advice in classical computational complexity theory~\cite{Karp1980-hr}.
For the most part, these examples seem slightly different than the quantum algorithmic primitives that we will discuss.
Classically, some applications of precomputation derive their usefulness from the ability to reuse the precomputed information rather than the time-sensitive nature of the computation.
In contrast with the classical case, the resource states that we consider are generally consumed when used, precluding their reuse.
It would be interesting if other techniques, perhaps based on gentle measurements~\cite{Aaronson2019-mk}, can be used to design quantum precomputation protocols that allow for some amount of information reuse.

\section{Examples of Precomputation}
\label{sec:precomputation_examples}

In this section, we discuss several examples of quantum precomputation.
These examples show how existing quantum primitives can be leveraged to obtain an advantage in a cost model that allows for free precomputation.
In particular, we study the application of density matrix exponentiation (introduced in \citen{Lloyd2014-td}, reviewed in \Cref{app:denmat_review}) and gate teleportation (introduced in \citen{Gottesman1999-gr}, reviewed in \Cref{app:review_gate_teleportation_clifford}) as tools for quantum precomputation.

Before turning towards these examples, it is worth briefly discussing two particularly simple forms of quantum precomputation.
One natural example is the case where precomputation is equivalent to performing the first steps of some algorithm and then waiting until the problem is fully specified to perform the rest.
For example, many quantum algorithms consist of applying a known unitary to the all zero state and performing a measurement.
If we knew the unitary ahead of time but the measurement wasn't yet specified, we could perform the state preparation in advance.
More speculatively, there may be settings where it is natural to prepare for the future execution of some quantum machine learning task by encoding data into a quantum state ``on the fly'' as it streams in.
This latter idea is related to rigorous work on quantum algorithms in streaming settings, which is itself connected to the study of quantum communication complexity~\cite{Le-Gall2006-nc, Kallaugher2022-nv}.

It is easy to understand how one might be able to usefully perform precomputation by executing the steps at the beginning of some algorithm ahead of time.
We could try to imagine situations where this naturally occurs, but it is unclear if our formal definition of the notion of quantum precomputation adds anything to the understanding of such cases.
For this reason, in the other examples that we consider in this paper, we focus instead on the goal of using precomputation to accelerate steps that lie in the middle of an algorithm, rather than at the beginning.

Turning towards a second example, recall that we briefly discussed the idea of precomputing magic states to use as resources for implementing non-Clifford gates in an error correcting code in \Cref{sec:prior_work}.
We explained how there is no advantage to this idea in the primary cost model we use throughout this paper because we do not distinguish between Clifford and non-Clifford gates.
This is true, but it is instructive to consider this example in a slightly different model of quantum precomputation, where we instead quantify the amount of spacetime volume required to implement a circuit in a quantum error correcting code.
For simplicity, let us work in units where a depth $d$ circuit acting on $n$ qubits has a volume of $dn$ and let us assume that the spacetime volume required to prepare a suitably distilled $T$ state is $\lambda \gg 1$.
Furthermore, we will neglect the spacetime cost of qubits that have not yet been initialized and qubits that have already been measured (since they could presumably be used for other purposes).

Under this more nuanced cost model, we can compare the cost of implementing an algorithm with and without the precomputed $T$ states.
Let us consider a depth $d$ circuit on $n$ qubits that consumes one magic state per time step.
Implementing this algorithm without precomputation would require a spacetime volume of $nd + \lambda d$ in order to account for the cost of the circuit itself and the cost of the magic state distillation.
In the precomputation model, we allow ourselves to start with all $d$ magic states already prepared, but we must account for the cost of storing them while the algorithm executes.
We are using $d - s$ qubits to store the magic states at each step $s$ from $0$ to $d-1$, so the spacetime volume required is 
$nd + \frac{d(d+1)}{2}$.

In this cost model, precomputing the magic states removes the dependence on $\lambda$ but it increases the dependence on $d$ from linear to quadratic.
Realistic values of $\lambda$ are expected to be significantly less than $100,$ which suggests that only relatively short-depth circuits of this type would benefit from free access to precomputed magic states~\cite{Litinski2019-ek}.
This example highlights the fact that our model implicitly penalizes precomputation protocols for the space used to store their precomputed resource states.
Because of this penalization, it is not trivially true that a precomputation protocol is at least as efficient as a straightforward approach to executing an algorithm.

\subsection{Precomputation with density matrix exponentiation}
\label{sec:density_matrix_exponentiation}

In this subsection, we consider applications where reflections about an expensive to prepare state, \(\ket{b}\), are a dominant contribution to the complexity of an algorithm.
As we explain below, an algorithm that requires \(q\) calls to the reflection operator \(R = \mathbb{I} - 2 \ketbra{b}\) can be implemented by consuming \(\bigo{q^2}\) copies of \(\ket{b}\) (at nearly unit time per consumption) in lieu of making any calls to \(R\) directly.
A cost model that allows for free precomputation can therefore entirely remove the component of such an algorithm's cost that depends on \(\ket{b}\).
In the most extreme cases, this could lead to a cost in the precomputation model that is exponentially smaller than the cost in a standard model.
For example, preparing or reflecting about the state \(\ket{b}\) might require using \(\text{poly}(|x|)\) gates to implement a brute-force encoding of some classical input \(x\) into \(n=\text{polylog}(|x|)\) qubits, while the other components of the algorithm could scale polynomially in \(n\).
We consider the quantum algorithm for linear systems as a specific example of an algorithm where such a speedup might prove useful~\cite{Harrow2009-iu,Childs2017-is}.

This type of quantum precomputation makes use of a technique called density matrix exponentiation.
Introduced in \citen{Lloyd2014-td}, density matrix exponentiation allows us to consume copies of some density matrix \(\rho\) in order to approximately apply the unitary \(e^{-it\rho}\) for some time \(t\).
We provide a brief review of density matrix exponentiation in \Cref{app:denmat_review}, but for now we just recall the fact that using density matrix exponentiation to implement \(e^{-it\rho}\) to within an error \(\epsilon\) (in the diamond norm) requires
\begin{equation}
    \label{eq:denmat_exp_scaling}
    m = \bigo{t^2/\epsilon}
\end{equation}
copies of \(\rho\)~\cite{Kimmel2017-af}.

Before explaining how we can make good use of density matrix exponentiation for quantum precomputation, let us examine why it does not lead to efficient protocols for implementing general unitaries in the precomputation cost model.
Imagine that we want to implement a unitary \(U\) that corresponds to evolution under a Hamiltonian \(H\) for a time \(t\), where \(||H||\) (the spectral norm of \(H\)) and \(t\) are both \(\bigo{1}\).
We can shift \(H\) by some multiple \(c\) of the identity to obtain a positive semidefinite operator \(H + c \mathbb{I}\) with \(||H + c \mathbb{I}|| = \bigo{1}\).
Applying \(U\) using density matrix exponentiation entails evolving under the Hamiltonian corresponding to the normalized state
\begin{equation}
    \label{eq:normalized_hamiltonian}
    \rho = \frac{H + c \mathbb{I}}{\tr \left[ H + c \mathbb{I} \right]}
\end{equation}
for a time
\begin{equation}
    \tilde{t} = t \tr \left[ H + c \mathbb{I} \right].
\end{equation}
The cost of implementing \(U\) using density matrix exponentiation scales quadratically with \(\tilde{t}\), which can scale exponentially with the number of qubits in the worst case.
This occurs easily even for simple unitaries, for example, when \(H\) is a non-trivial Pauli operator.

In order for density matrix exponentiation to be a useful tool for precomputation, we need to focus on cases where the normalization factor is small.
One natural example of a unitary that is efficiently implementable using density matrix exponentiation is the reflection about a state \(\ket{b}\),
\begin{equation}
    R = \mathbb{I} - 2 \ketbra{b} = e^{-i\pi \ketbra{b}}.
\end{equation}
In order to implement \(R\) up to an accuracy \(\tilde{\epsilon}\) using density matrix exponentiation, it suffices to consume \(\bigo{\tilde{\epsilon}^{-1}}\) copies of the state \(\ketbra{b}\).
If an algorithm involves \(q\) calls to \(R\), we can guarantee a constant overall error \(\epsilon\) by setting \(\tilde{\epsilon} \propto \epsilon q^{-1} \).
We can therefore implement all \(q\) calls to \(R\) to within the desired accuracy by consuming a total of \(\bigo{\epsilon^{-1} q^2}\) copies of \(\ketbra{b}\).

As an example of a context where this kind of precomputation might be useful, consider the quantum linear systems problem~\cite{Harrow2009-iu, Childs2017-is, Lin2020-ib, Somma2021-ql, Martyn2021-mf, Costa2022-ae}.
Given a matrix \(A\) and a vector \(\vec{b}\), the linear systems problem is to find a vector \(\vec{x}\)
such that \(A \vec{x} = \vec{b}\).
The quantum formulation of this problem encodes the vector \(\vec{b}\) into the amplitudes of a state \(\ket{b}\) and asks that we prepare a state \(\ket{x} \propto A^{-1} \ket{b}\).
Without loss of generality we can assume that \(A\) is Hermitian.\footnote{One can always solve a linear systems problem on a larger space with the Hermitian
    \(
    \tilde{A} \coloneqq
    \begin{bmatrix}
        0         & A
        \\
        A^\dagger & 0
    \end{bmatrix}
    \)
    instead of the original \(A\).}
The access models for \(A\) and \(\ket{b}\) can vary, but it is usually assumed that one has access to an oracle that prepares \(\ket{b}\) and either i) the ability to perform time evolution by \(A\), ii) oracle access to the non-zero entries of (a sparse) \(A\), or iii) a block encoding of \(A\). 
Regardless of the access model for \(A\), the most efficient algorithms for this problem query the state preparation oracle for \(\ket{b}\) a number of times that scales as \(\bigot{\kappa}\), where \(\kappa\) denotes the condition number of \(A\) and the \(\bigot{\cdot}\) notation hides logarithmic factors in \(\kappa\) and the precision.
These queries are used to prepare \(\ket{b}\) and to implement the reflection \(R\) about \(\ket{b}\).

In a context where a classical description of \(\ket{b}\) is available before \(A\), preparing \(\bigot{\epsilon^{{-1}} \kappa^2}\) copies of \(\ket{b}\) during the precomputation step would allow us to apply one of the standard quantum algorithms for the linear systems problem at a cost that is independent of the cost of preparing \(\ket{b}\).
As we argued above, it is easy to imagine situations where preparing or reflecting about \(\ket{b}\) is exponentially more expensive than any other component of the algorithm.
For example, we could take \(\ket{b}\) to be a brute force encoding of some classical data \(|x|\) into \(n = \text{polylog}(|x|)\) qubits, such that preparing or reflecting about \(\ket{b}\) has a complexity that scales polynomially in \(|x|\).
We could also make the (sometimes reasonable) assumption that the condition number of \(A\) and the gate complexity of implementing \(A\) (under whatever notion of access is appropriate) scale polynomially in \(n\).
Given these two conditions, the complexity of applying any of the standard quantum algorithms for the linear systems problem would be exponentially better in the precomputation cost model than in the standard one (assuming that the target precision is a constant).

Of course, this separation is entirely due to the fact that we discount the cost of preparing the resource state.
In fact, in this form of precomputation, the cost of preparing the resource state would be asymptotically larger than the cost of implementing the reflections in the standard way since we require \(\bigot{q^2}\) copies of \(\ket{b}\) to implement the reflection \(R\) a total of \(q\) times with constant error in the overall algorithm.
Furthermore, the optimal algorithms for the quantum linear systems problem have a logarithmic dependence on the target precision~\cite{Childs2017-is}, whereas our approach introduces a polynomial dependence.
Additionally, sufficient storage for the copies of \(\ket{b}\) would be required.
Nevertheless, in a situation where \(\vec{b}\) is specified ahead of time and the solution to the problem is sufficiently valuable and time-sensitive, quantum precomputation could prove useful.
Note that there is no significant classical cost in terms of storage or computation for this form of precomputation.

It is worth point out that, if one is willing to prepare \(\bigot{\kappa^2}\) copies of \(\ket{b}\) ahead of time, there is a simpler strategy to solving the linear systems problem that does not require density matrix exponentiation.
However, this simpler strategy is less efficient with respect to the number of times that \(A\) must be queried.
Consider the original HHL algorithm of \citen{Harrow2009-iu}.
This algorithm requires starting with the state \(\ket{b}\) and time-evolving under the Hamiltonian \(A\) for a time that scales as \(\bigot{\kappa}\) (to perform phase estimation).
This is followed by a postselection step that succeeds with probability \(\bigomega{1/\kappa^2}\).
Normally one uses amplitude amplification to increase the success probability to \(\bigo{1}\).

Instead of using amplitude amplification, one could instead repeatedly prepare the appropriate state and actually perform the postselection based on the output from phase estimation.
This would solve the quantum linear systems problem with high probability using a number of copies of \(\ket{b}\) that scales as \(\bigot{\kappa^2}\).
However, it would also require a total amount of time evolution under \(A\) equal to \(\bigot{\kappa^3}\).
The approach we proposed above uses a similar number of copies of \(\ket{b}\), but the scaling in terms of \(A\) (either time evolution under \(A\) or a related notion of access) can be made nearly linear with respect to \(\kappa\) by using the optimal algorithms of, e.g., \citen{Childs2017-is}.

\subsection{Precomputing Clifford unitaries with gate teleportation}
\label{sec:gate_teleportation_precomputation}

In this subsection, we consider accelerating the task of implementing an \(n\)-qubit unitary from the Clifford group using precomputation.
We explain how a well-known construction allows for a quadratic savings in gate complexity (when comparing the cost in the precomputation model to the gate complexity in a standard cost model).
This construction is a straightforward application of gate teleportation, a technique introduced in \citen{Gottesman1999-gr} which we illustrate in \Cref{fig:generic_teleportation} and review in more detail in \Cref{app:review_gate_teleportation_clifford} (along with the definition of the Clifford group).
Although this example of quantum precomputation is particularly simple, it provides a good introduction to some of the concerns relevant in the more technically interesting example that we consider in \Cref{sec:selective_gate_teleportation_precomputation}.

\begin{figure}[htbp]
    \centering
    \begin{quantikz}[font=\footnotesize, column sep = 3.4mm]
        \lstick{\(\ket{\psi}\)} & \qw                                                                                                                                                                                    & \qw      & \qw            & \qw & \ctrl{1} \gategroup[2,steps=2,style={dashed,rounded corners,fill=red!20, inner sep=0pt}, background]{{Bell basis measurement}} & \meterD{X} & \cw & \cw \cwbend{1}
        \\
        \lstick{$\ket{0}$}        & \gate{H} \gategroup[2,steps=2,style={dashed,rounded corners,fill=blue!20, inner sep=0pt}, background,label style={label position=below,anchor=north,yshift=-0.2cm}]{{Bell state prep}} & \ctrl{1} & \qw & \qw & \targ{}                                                                                                                        & \meterD{Z} & \cw          & \cw \cwbend{1}
        \\
        \lstick{$\ket{0}$}        & \qw                                                                                                                                                                                    & \targ{}  & \gate{U}       & \qw & \qw                                                                                                                            & \qw                      & \qw          & \gate{U P^\dagger U^\dagger} & \qw \rstick{\(U \ket{\psi}\)}
    \end{quantikz}
    \caption{A quantum circuit diagram for the one-qubit version of gate teleportation~\cite{Gottesman1999-gr}. The circuit in the blue shaded area prepares a bell pair and the circuit in the red shaded area performs a bell basis measurement (the \(X/Z\) in the rounded caps indicate \(X/Z\) basis measurements). Based on the outcome of the measurement, a classically controlled operation \(U P^\dagger U^\dagger\) is performed, where the ``byproduct operator'' \(P \in \left\{ \mathbb{I}, X, Z, ZX \right\}\) depends on the measurement outcome. When \(U\) is a member of the Clifford group, \(U P^\dagger U^\dagger\) is an element of the Pauli group. Single-qubit gate teleportation generalizes naturally to a multi-qubit version. Using multi-qubit gate teleportation to apply an \(n\)-qubit unitary from the Clifford group offers a simple example of advantage in the precomputation cost model, reducing the quantum gate complexity from \(\bigo{n^2}\) to \(\bigo{n}\).}
    \label{fig:generic_teleportation}
\end{figure}
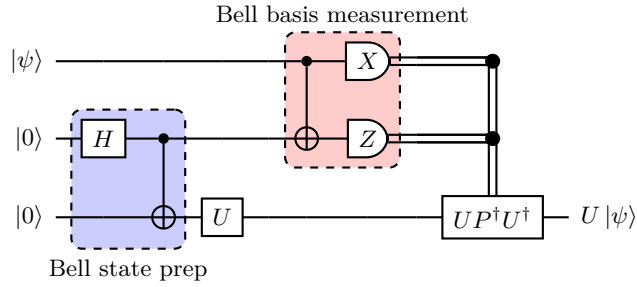

We recall that an arbitrary unitary from the \(n\)-qubit Clifford group can be efficiently implemented using one- and two-qubit Clifford gates arranged in a circuit with depth \(\bigo{n}\)~\cite{Bravyi2021-px}, leading to a gate complexity of \(\bigo{n^2}\).
A counting argument shows that this asymptotic scaling must be optimal for most elements of the Clifford group.
We will show that, in the precomputation cost model, the quantum gate complexity of applying the same unitary is only \(\bigo{n}\).

Let \(U\) be an arbitrary unitary in \(\mathcal{C}^{(2)}\) (the Clifford group on \(n\) qubits) and \(\ket{\psi}\) be an arbitrary \(n\)-qubit quantum state.
Using standard multi-qubit gate teleportation, we can prepare a state \(\ket{\Gamma(U)}\) on \(2n\) qubits that we can consume to apply \(U\) to \(\ket{\psi}\) (up to a Pauli correction).
This straightforward generalization of the procedure presented in \Cref{fig:generic_teleportation} consists of preparing \(n\) bell pairs and applying \(U\) to a set of \(n\) qubits, one taken from each bell pair.
Let us consider the steps involved in applying \(U\) to \(\ket{\psi}\) once \(\ket{\Gamma(U)}\) is already prepared.
Applying a Clifford unitary using gate teleportation involves making \(n\) simultaneous bell-basis measurements of the \(3n\)-qubit state \(\ket{\psi} \otimes \ket{\Gamma(U)}\).
The resulting \(n\)-qubit state can therefore be obtained in constant depth,
\begin{equation}
    \ket{\phi} = U P \ket{\psi},
\end{equation}
where \(P\) is the ``byproduct operator,'' a member of the Pauli group that is determined by the measurement outcomes.
By the definition of the Clifford group, the correction operator \(U P^\dagger U^\dagger\) is also a Pauli operator (up to a possible phase) and can therefore be applied in constant depth to yield the desired state \(U \ket{\psi}\).
The overall quantum circuit complexity (neglecting the cost of preparing \(\ket{\Gamma(U)}\)) is therefore \(\bigo{n}\), in contrast with the \(\bigo{n^2}\) cost of applying \(U\) without precomputation.

Although we are primarily concerned with the quantum gate complexity of applying \(U\) given \(\ket{\Gamma(U)}\), we may also wish to consider the classical computational costs of determining which of the \(4^n\) possible correction operators to apply once the measurement outcomes are known.
We need to use \(2n\) bits initially to store the results of the bell basis measurement that determines the byproduct operator.
We could store a classical description of the \(\mathcal{O}(n^2)\) Clifford gates in \(U\) and apply them to the byproduct operator.
This would require \(\mathcal{O}(n^2)\) operations (updating a constant number of the \(\bigo{n}\) stored bits each time we conjugate by a gate in the circuit) which could be performed in \(\bigo{n}\) sequential steps by parallelizing across gates in the same layer of the circuit.

We can reduce the depth of the classical computation (although not the overall number of operations) by factorizing the correction operator ahead of time,
\begin{equation}
    \label{eq:pauli_correction_operator_factorization}
    U P^\dagger U^\dagger
    =               U
    \left(
    \bigotimes_{i=1}^n X_i^{x_i} Z_i^{z_i}
    \right)
    U^\dagger
    =              
    \left(
    \prod_{i=1}^n U X_i^{x_i} U^\dagger
    \right)
    \left(
    \prod_{i=1}^n U Z_i^{z_i} U^\dagger
    \right)
    ,
\end{equation}
where the \(x_i\) and \(z_i\) are determined by the measurement outcomes of bell basis measurement.
This allows us to classically precompute each of the \(2n\) Pauli operators of the form \(U X_i U^\dagger\) or \(U Z_i U^\dagger\)
and store the results using \(\bigo{n^2}\) bits.
Once we know the measurement results, we can multiply the appropriate operators together in logarithmic depth using a divide and conquer strategy, ultimately computing the final correction operator using \(\mathcal{O}(n^2)\) operations using \(\bigo{\log(n)}\) sequential steps (neglecting the classical cost of the precomputation).

\section{Precomputing diagonal unitaries in the Clifford hierarchy with selective gate teleportation}
\label{sec:selective_gate_teleportation_precomputation}

In this section, we show how a more sophisticated form of gate teleportation introduced in~\citen{Fowler2012-ti} can be used to construct a precomputation protocol for a set of diagonal unitaries in the Clifford hierarchy (reviewed \Cref{app:review_gate_teleportation_clifford}).
We graphically illustrate this selective gate teleportation in \Cref{fig:selective_gate_teleportation} and present a more substantial review in \Cref{app:selective_teleportation}.
In \Cref{sec:gate_teleportation_precomputation}, we considered a simple example of quantum precomputation that uses standard gate teleportation to apply some \(U \in \mathcal{C}^{(2)}\) (the Clifford group).
We explained how the \(\bigo{n^2}\) gate complexity required to implement an arbitrary \(n\)-qubit Clifford unitary can be reduced to \(\bigo{n}\) in the precomputation model.
The approach is less straightforward, but the generalization that we present in this section achieves the same quadratic compression for a subset of unitaries from higher levels of the Clifford hierarchy.
In other words, we show that unitaries from the family \(\mathcal{Z}^{(k)}\), defined below, that have a gate complexity of \(\bigtt{n^k}\) when implemented directly can be implemented with a gate complexity of \(\bigot{k n^{k/2}}\) in the precomputation cost model (assuming \(k\) is even for simplicity).
The basic strategy we use is to apply such a unitary with gate teleportation and then use a series of selective gate teleportation steps to apply the correction operator up to some simpler correction that can be implemented directly.

Before we present our actual proposal, let us consider a naive generalization, where we use gate teleportation to implement some \(U \in \mathcal{C}^{(3)}\) (the third level of the Clifford hierarchy).
By definition, the correction operator required will be some \(R \in \mathcal{C}^{(2)}\).
Applying \(R\) directly would result in an overall gate complexity of \(\bigo{n^2}\), essentially saving a factor of \(n\) compared to the cost of implementing \(U\) directly, which is \(\bigomega{n^3}\) by a counting argument.
For a general \(U \in \mathcal{C}^{(k)}\), it is unclear if it is possible to obtain an advantage greater than a factor of \(n\) in the precomputation model.

However, if we restrict ourselves to considering a smaller set of unitaries, we can do better.
Rather than allowing for arbitrary elements of the Clifford hierarchy, we limit ourselves to considering elements of the hierarchy that are also diagonal.
To simplify the presentation, we actually restrict ourselves even further in this section, considering only those gates in
\(\mathcal{C}^{(k)}\) that are composed of products of \(\pm \mathbb{I}\), Pauli \(Z\) operators, and controlled \(Z\) operators with up to \(k - 1\) controls.\footnote{The only property of \(\mathcal{Z}^{(k)}\) that we leverage, other than the fact that \(\mathcal{Z}^{(k)} \subset \mathcal{C}^{(k)}\), is that it forms an Abelian group. This is also true of the full set of diagonal unitaries at each level in the Clifford hierarchy, which suggests that our results may readily generalize to this case.}
We denote this set \(\mathcal{Z}^{(k)}\) and in \Cref{app:Zk_hierarchy}, we show that it forms a group.
We also note that \(\mathcal{Z}^{(j)}\) is a proper subgroup of \(\mathcal{Z}^{(k)}\) for \(j < k\) and prove the following proposition:
\begin{restatable}[]{proposition}{GXcommutation}
    \label{proposition:GX_commutation}
    Consider a gate \(G \in \mathcal{Z}^{(k)}\) and a product of single-qubit Pauli \(X\) operators that we denote by \(X_{\boldsymbol{s}}\) (where \(\boldsymbol{s} \in [n]\) indicates the indices of the qubits where \(X_{\boldsymbol{s}}\) acts non-trivially).
    Define \(G'\) in the following way,
    \begin{equation}
        G' \coloneqq X_{\boldsymbol{s}} G X_{\boldsymbol{s}} G^\dagger.
    \end{equation}
    Then \(G' \in \mathcal{Z}^{(k - 1)}\) if \(k > 1\) and \(G' = \pm \mathbb{I}\) if \(k \in \left\{ 0, 1 \right\}\).
    As a corollary, we also have that
    \begin{equation}
        G X_{\boldsymbol{s}} = X_{\boldsymbol{s}} G' G.
    \end{equation}
\end{restatable}

Diagonal unitaries commute, and the elements of \(\mathcal{Z}^{(k)}\) are all self-inverse.
As a result, we can specify a \(U \in \mathcal{Z}^{(k)}\) using exactly
\begin{equation}
    \sum_{j=0}^k \binom{n}{j} = \bigo{n^k}
\end{equation}
bits,  
one to specify the sign and one to specify the presence or absence of each possible \(C^{j-1}Z\) gate for each \(j \in [1..n]\).
A \(C^{j-1}Z\) gate can be implemented using \(\bigo{j}\) \(T\) gates in depth \(\bigo{\log j}\)~\cite{Motzoi2017-rd}.
An arbitrary gate \(G \in \mathcal{Z}^{(k)}\) can therefore be implemented in depth \( \bigot{n^{k-1}}\) and gate complexity \(\bigot{n^k}\), even under reasonable assumptions about qubit connectivity~\cite{O-Gorman2019-dj}.
Furthermore, by counting the number of distinct elements of \(\mathcal{Z}^{(k)}\), we can also see that a typical element must have a circuit complexity lower bounded by \(\bigomega{n^k}\).

\begin{figure}[tbp]
\centering
    \begin{quantikz}
        [font=\footnotesize, column sep = 3.4mm, execute at end picture={
        \node (A)[fit=(\tikzcdmatrixname-1-10)(\tikzcdmatrixname-1-11)(\tikzcdmatrixname-4-10)(\tikzcdmatrixname-4-11), inner sep=4pt, rounded corners] {};
        \draw [->, dotted, thick,-{Latex[round]}] (A.south) -- ++(0,-.4);
        \node (B)[fit=(\tikzcdmatrixname-1-12)(\tikzcdmatrixname-1-13)(\tikzcdmatrixname-4-12)(\tikzcdmatrixname-4-13), inner sep=4pt, rounded corners] {};
        \draw [->, dotted, thick,-{Latex[round]}] (B.south) -- ++(0,-.4);
        }]
        \lstick{$\ket{\psi}$}                                                & \ctrl{1}   & \targ{}                            & \qw       & \qw      & \qw      & \qw      & \qw \arrow[rr, thick,-{Latex[round]}, dotted, shorten >=7pt]     &     & \gategroup[wires=4,steps=2,style={
        dashed,rounded corners,fill=blue!20, inner sep=0pt},background]{{A}} & \meterD{Z} & \gategroup[wires=4,steps=2,style={
        dashed,rounded corners,fill=red!20, inner sep=0pt},background]{{B}}  & \meterD{X}
        \\
        \lstick{$\ket{0}$}                                                   & \targ{}    & \qw                                & \targ{}   & \qw      & \qw      & \qw      & \qw \arrow[rr, thick,-{Latex[round]}, dotted, shorten >=7pt]     &     &                                    & \meterD{X} &     & \meterD{Z}
        \\ [2pt]
        \lstick{$\ket{+}$}                                                   & \qw        & \ctrl{-2}                          & \qw       & \gate{U_A}     & \ctrl{2} & \qw      & \qw   \arrow[rr, thick,-{Latex[round]}, dotted, shorten >=7pt]   &     &                                    & \meterD{X} &     & \meterD{Z}
        \\
        \lstick{$\ket{+}$}                                                   & \qw        & \qw                                & \ctrl{-2} & \gate{U_B} & \qw      & \ctrl{1} & \qw     \arrow[rr, thick,-{Latex[round]}, dotted, shorten >=7pt] &     &                                    & \meterD{Z} &     & \meterD{X}
        \\
        \lstick{$\ket{0}$}                                                   & \qw        & \qw                                & \qw       & \qw      & \targ{}  & \targ{}  & \qw                                                              & \qw & \qw                                & \qw        & \qw & \qw \rstick{\({\color{Blue3}
        P^{(1)} U_A P^{(2)}\ket{\psi}} \big/ {\color{Red3} P^{(1)} U_B P^{(2)} \ket{\psi}}\)}
    \end{quantikz}
    \caption{A circuit diagram for the one-qubit version of selective gate teleportation~\cite{Fowler2012-ti}. This protocol allows for the teleportation of a choice of unitaries, \(U_A\) or \(U_B\), onto an input state. Which unitary is teleported is controlled by the measurement settings (the four ancilla qubits are each measured in the \(X\) or \(Z\) basis according to the proscriptions shown in the blue and red shaded areas of the diagram). 
    The possible states of the output qubit are color-coded to match the measurement settings that select for them.
    In our use of selective teleportation, we take \(U_A = U\) and \(U_B = \mathbb{I}\).
    Byproduct operators \(P^{(1)}\) and \(P^{(2)}\) from the set \(\left\{ \mathbb{I}, X, Z, XZ \right\}\) are randomly applied before and after the selected unitary based on the measurement outcomes.
    }
    \label{fig:selective_gate_teleportation}
\end{figure}
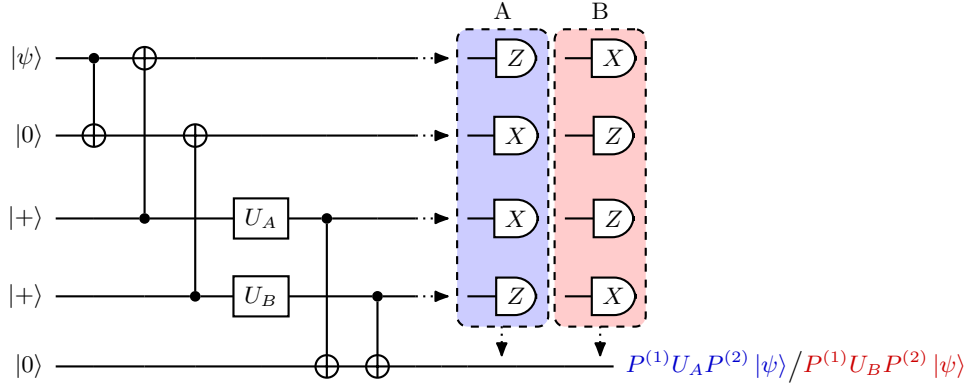

We begin our construction by preparing the usual \(2n\) qubit resource state for applying the gate \(U \in \mathcal{Z}^{(k)}\) using teleportation.
If this state were used directly for gate teleportation, we would need to perform a correction of the form \(U P^\dagger U^\dagger\) for some \(n\)-qubit byproduct operator \(P\) (which we can write as a product of single-qubit \(X\) and \(Z\) operators).
We will perform this correction using selective teleportation.
Note that we can neglect the \(Z\) corrections (as they can be trivially commuted to the end of the circuit up to a sign).
Factorizing the \(X\) component of the corrections, we see that we need to apply the unitary
\begin{equation}
    R = \prod_{i=1}^{n} U X_i^{x_i} U^\dagger,
\end{equation}
where the bits \(x_i\) will be chosen based on the measurement outcomes of first gate teleportation step.
It is convenient to rewrite each of the terms in the product as
\begin{equation}
    \label{eq:diagonal_factored_correction_operator}
  U X_i^{x_i} U^\dagger = X_i^{x_i} \left( X_i^{x_i} U X_i^{x_i} U^\dagger \right),
\end{equation}
i.e., a product of \(X_i^{x_i}\) and an operator that is in \(\mathcal{Z}^{(k-1)}\) by \Cref{proposition:GX_commutation}.

We can use selective gate teleportation to apply the diagonal term (\( X_i^{x_i} U X_i^{x_i} U^\dagger\)) from each of the \(n\) possible factors of the correction operator.
Note that we can do this after applying \(U\) to the \(n\) bell pairs and before performing the bell basis measurement that completes the gate teleportation.
We ignore the \(X_i^{x_i}\) terms that precede the diagonal components of the factors of the correction operator in \Cref{eq:diagonal_factored_correction_operator} because we can absorb them into the byproduct operators that will arise anyway from the selective teleportation.
For each of the correction operators, we need \(4n\) additional qubits to implement the selective gate teleportation, so the overall overhead is \(4n^2\).
When we attempt to use selective teleportation in this way to implement the correction operator, we will actually end up implementing the operator
\begin{equation}
    \tilde{R} = P^{(0)} \prod_{i=1}^{n} \left( X_i^{x_i} U X_i^{x_i} U^\dagger P^{(i)}\right),
\end{equation}
where the \(P^{(i)}\) terms represent randomly obtained products of Pauli operators and the \(X_i^{x_i} U X_i^{x_i} U^\dagger\) are elements of \(\mathcal{Z}^{(k-1)} \).
Notice that we can commute the Pauli terms to the left at the cost of requiring a series of corrections \(R^{(i)'} \in \mathcal{Z}^{(k-2)}\).

We can proceed recursively.
We factored the one byproduct operator to obtain \(n\) possible factors of the correction operator, each of which we applied using selective gate teleportation.
Implementing these corrections required a total of \(4n^2\) additional ancilla qubits and resulted in the addition of Pauli byproduct operators at \(n + 1\) locations.
We can commute these byproduct operators through to the left, starting at the righthand side of our expression. 
Each time we commute an \(n\)-qubit operator of the form \(\prod_{i=1}^n X_i^{x_i}\) through a diagonal gate we do so by factorizing it and we pick up \(n\) possible correction terms one level lower in the \(\mathcal{Z}^{(k)}\) hierarchy.
The number of corrections that we must apply, and the number of additional ancilla qubits that we require, therefore increases by a factor of \(n\) each time we descend the hierarchy by a level.
For example, we can use \(\bigo{n^3}\) ancilla qubits to implement each of the \(n^2\) possible second-order corrections using selective teleportation, leaving only corrections that are three or more levels down the hierarchy.
More generally, to implement \(U \in \mathcal{Z}^{(k)} \) up to a correction \(R \in \mathcal{Z}^{(k-a)}\) (and some Pauli \(X\) operators), we require a resource state on \(\bigo{n^a}\) qubits.

If we were to descend the hierarchy all the way to the point where the only remaining corrections were Pauli corrections (\(a = k - 1\)), we would obtain only a modest compression in circuit complexity (compared with directly applying \(U\)).
This is because, although the circuit depth would be merely \(\bigo{k}\), we would require \(\bigo{n^{k-1}}\) qubits.
However, consider what happens when we stop at the level \(a = \lfloor k/2 \rfloor \).
To simplify the presentation we assume that \(k\) is even.
We can use a resource state on \(\bigo{n^{k/2}}\) qubits to implement \(U\) up to a correction \(R \in \mathcal{Z}^{(k/2)}\) (and some additional Pauli terms) in \(k/2\) rounds of measurement.
We can implement the remaining correction directly with a gate complexity of \(\bigot{kn^{k/2}}\) in depth \(\bigot{k}\) with no additional space overhead using the constant depth fanout and unfanout circuits of \citen{Pham2012-rb}.
Therefore, the overall gate complexity of implementing an arbitrary \(U \in \mathcal{Z}^{(k)}\) in the precomputation model (i.e., neglecting the cost of preparing the resource state) is \(\bigot{k n ^{k/2}}\).

Recall that a fanout operation takes an \(n\)-qubit state \(\ket{\psi}\) and performs the map
\begin{equation}
    \ket{\psi} = \sum_{i=1}^{2^n} c_i \ket{i} \rightarrow \sum_{i=1}^{2^n} c_i \ket{i}^{\otimes m}
\end{equation}
for some integer \(m > 1\), where the states in \(\left\{ \ket{i} \right\}\) are the computational basis states.
Unfanout reverses this mapping.
\citen{Pham2012-rb} explains how both of these operations can be implemented using constant depth quantum circuits and classical feedback.
We can parallelize the implementation of \(m\) diagonal unitaries by performing a fanout, applying each unitary to a separate fanned out copy of \(\ket{\psi}\), and then performing an unfanout.

We can take advantage of this capability by partitioning the individual terms that make up an arbitrary \(R \in \mathcal{Z}^{(k/2)}\) into \(\bigo{k}\) sets of gates, where each set contains only terms that act on disjoint qubits.
By setting \(m = n^{k/2 - 1}\), we can apply the terms from each of the sets in parallel.
We can therefore apply all of the terms with the desired gate complexity and depth. Because the fanout and unfanout operations are constant depth, they do not increase the asymptotic scaling of the gate complexity.
The remaining Pauli correction can then be applied to complete the implementation of \(U\).

Now let us consider the classical computational cost associated with applying \(U\) this way in the precomputation model.
Applying \(U\) up to a correction at level \(\mathcal{Z}^{k-a}\) is trivial for \(a=1\).
For \(a=2\), we apply some subset of the \(n\) possible corrections that corresponds directly to the bits we obtained from the first set of measurements.
For \(a=3\), we need to repeatedly \(\text{XOR}\) one \(n\) bit string into another \(\bigo{n}\) times in order to determine the measurement settings, using \(\bigo{n^2}\) classical operations.
This growth continues, and we find that we need to perform \(\bigo{n^{k/2-1}}\) classical operations to determine which corrections to perform at the level that leaves us with a final correction in \(\mathcal{Z}^{k/2}\).
Actually computing the final correction \(R \in \mathcal{Z}^{k/2}\) requires determining the \(\bigo{n^{k/2}}\) elements of \(\mathcal{Z}^{k/2}\) that arise from commuting the byproduct operators through and then taking their product, which ultimately takes \(\bigo{n^k}\) \(\text{XOR}\) operations.
The classical postprocessing involved in the fanout operation is negligible compared to these costs, so the overall classical complexity is \(\bigo{n^k}\).

We can also ask about the quantum and classical complexities of performing the precomputation step.
Neglecting the operations involved in setting up the teleportation and selective teleportation gadgets themselves since they contribute negligibly to the overall complexity, we can just consider the gate complexities of performing one operation from \(\mathcal{Z}^{(k)}\), \(n\) operations from \(\mathcal{Z}^{(k-1)}\), and so on, down to \(n^{k/2-1}\) operations at the level \(\mathcal{Z}^{(k/2 + 1)}\). 
The only clear way to apply these operations is to work serially (since the use of selective teleportation may prevent us from using fanout and unfanout operations to parallelize).
This means that, although we only require \(\bigot{k n^k}\) non-identity gates, our definition of gate complexity (which attempts to account for storage space by counting the single-qubit identity operation as a gate) implies that the overall gate complexity of the precomputation step is \(\bigot{k n^{3k/2}}\).
This may not be a fundamental requirement, and it is also true that most of the \(\bigo{n^{k/2}}\) qubits are not required at all until the very last portions of the precomputation step, so they could be used for other things in the meantime.
The classical complexity of the precomputation step arises from computing the various correction operators and is not substantially larger than would be expected from the need to generate some kind of classical description of the circuits involved anyway.

In many ways, the techniques of this section are a generalization of the simpler scheme for applying Clifford operators using gate teleportation that we presented in \Cref{sec:gate_teleportation_precomputation}.
In order to make a comparison easy, we summarize the various scalings of these two examples of quantum precomputation in \Cref{tab:my-cost-table}.
\begin{table}[htbp]
    \centering
    \footnotesize
    \begin{tabular}{@{}lll@{}}
    \toprule
     &
      Typical \(U \in \mathcal{C}^{(2)}\) (Sec.~\ref{sec:gate_teleportation_precomputation}) \;\; &
      Typical \(U \in \mathcal{Z}^{(k)}\) (Sec.~\ref{sec:selective_gate_teleportation_precomputation}) \;\; \\ \midrule
    Gate complexity, standard                              & \(\bigt{n^2}\)  & \(\bigtt{n^k}\)     \\
    Gate complexity, precomputation \;\;              & \(\bigo{n}\)    & \(\bigot{k n^{k/2}}\) \\
    Resource state size                                    & \(\bigo{n}\)    & \(\bigo{n^{k/2}}\) \\
    Gate complexity, preparing \(\ket{\Gamma(U)}\)         & \(\bigo{n^2}\)  & \(\bigot{k n^{3k/2}}\)     \\
    Classical operations, consuming \(\ket{\Gamma(U)}\)\;\; & \(\bigo{n^2}\) & \(\bigo{n^k}\)     \\ \bottomrule
    \end{tabular}
    \caption{A summary of the scalings for applying arbitrary Clifford operators using gate teleportation (\Cref{sec:gate_teleportation_precomputation}) and arbitrary elements of \(\mathcal{Z}^{k}\) (products of \(Z\) and controlled \(Z\) operators with up to \(k-1\) controls) using selective gate teleportation (\Cref{sec:selective_gate_teleportation_precomputation}). For simplicity we assume that \(k\) is even. For the gate complexity, we count the number of one- and two-qubit gates from the Clifford + \(T\) gate set (counting single-qubit identity operations as gates). The quoted gate complexity in the precomputation model includes only those quantum operations required to consume the resource state \(\ket{\Gamma(U)}\). The (quantum) cost of preparing the resource state is provided separately, as is the number of classical operations required to consume the resource state to apply \(U\).}
    \label{tab:my-cost-table}
    \end{table}

\section{Discussion}
\label{sec:discussion}

In this paper, we introduced a new cost model for quantum computation that allows for ``quantum precomputation.''
This model is motivated by practical scenarios where it is highly valuable to perform a time-sensitive computation as quickly as possible, and where some portion of the problem's input is naturally known ahead of time.
In the precomputation cost model, we allow a reasonable (polynomial in the input size) amount of effort to be spent ``for free'' preparing a resource state before the input is fully specified.
The cost of an algorithm in the precomputation cost model is determined solely by the resources required to implement the algorithm given access to the resource state.
We presented three realizations of quantum precomputation that require asymptotically fewer resources in the precomputation cost model than in a standard one.

The first realization uses density matrix exponentiation to implement reflections about a state by consuming copies of that state.
We explained how, in some cases, this type of quantum precomputation can offer an exponential advantage (in the sense that the complexity required to execute an algorithm by consuming the resource state can be exponentially smaller than the complexity required to execute an algorithm directly).
As a particular example, we considered the task of accelerating quantum algorithms for linear systems in cases where it is natural to prepare copies of the state \(\ket{b}\) ahead of time.
In the future, we hope to find practical examples where this type of precomputation is useful, either for solving particular linear systems of equations, or for executing some other quantum algorithm whose cost might be dominated by the cost of implementing low-rank reflections.
In practice, the advantage need not be exponential to be useful.
It would be especially interesting if we could find situations where the ability to accelerate an algorithm using precomputation was the deciding factor that made it worth solving a particular problem using quantum rather than classical computation.

As a second example, we pointed out that standard techniques for implementing Clifford unitaries using gate teleportation constitute a simple illustration of an asymptotic advantage in the precomputation cost model.
These techniques allow for unitaries with a gate complexity of \(\bigt{n^2}\) to be implemented in \(\bigo{1}\) (quantum gate) depth by consuming a state on \(2n\) qubits.
This example highlights the importance of choosing an appropriate notion of cost when defining a model of quantum precomputation.
Under a definition of cost that treated Clifford operations as free, there could be no value in using precomputation to apply a Clifford unitary more efficiently.
However, as schemes for magic state distillation continue to improve, it is becoming less clear if quantifying the cost of a fault-tolerant quantum algorithm solely in terms of the number of non-Clifford gates is an accurate approximation~\cite{Litinski2019-ek}.
This motivated our particular definition of a precomputation cost model (that counts gate complexity, including Clifford gates), but it is possible that a metric of cost even closer to the hardware might be more appropriate.
For instance, one could imagine squeezing some additional benefit out of a scheme for quantum precomputation by preparing the resource states using shorter distance error correcting codes (and therefore, fewer physical qubits and less actual time) in conjunction with error detection and postselection.

Even within the particular cost model we have defined, there are many degrees of freedom to explore in defining precomputation protocols.
For example, the technique we used to implement an arbitrary Clifford unitary could be modified to apply an \(n\)-qubit circuit \(U\) that interleaved Clifford operations with a small number \((t)\) of \(T\) gates.
Such a modified scheme could use a combination of gate teleportation and selective teleportation to apply the Clifford gates as normal, while selectively implement the possible corrections after each \(T\) gate. 
This would require an \(\bigo{n + t}\) qubit resource state that would be consumed in \(\bigo{t}\) rounds of measurement to apply \(U\) up to a final Pauli correction.

The most novel example of precomputation that we proposed in this paper uses selective teleportation to achieve a quadratic reduction in the complexity of implementing a family of diagonal unitaries from the Clifford hierarchy (when comparing the cost in the precomputation model with the standard cost).
Our scheme is likely generalizable to all diagonal unitaries that are members of the Clifford hierarchy, but this is still a relatively restricted class of unitaries.
This naturally raises the question, are there ways to compile existing algorithms such that they would make heavy use of the kinds of diagonal unitaries that we have shown can be accelerated by precomputation?
Diagonal unitaries appear in a variety of places, oftentimes as a natural way of encoding the output of a classical function into a phase.
For example, the Forrelation problem~\cite{Aaronson2015-zu}, IQP circuits~\cite{Shepherd2009-rw}, QAOA~\cite{Farhi2014-vk}, and Grover's algorithm itself~\cite{Grover1996-ah}, can all be formulated to involve heavy use of diagonal unitaries.
In the future, we hope that extensions of our precomputation protocols can be used to accelerate some such algorithms for interesting and time-sensitive applications.

More broadly, does quantum precomputation have anything to teach us about the nature or power of quantum computation?
The power of advice (computation supplemented by a resource state) has been studied both in classical and quantum contexts~\cite{Karp1980-hr, Aaronson2004-xr}, but, as we discuss in \Cref{sec:precomputation}, the precomputation model we introduced differs from these prior works in that we require that the extra resource state be efficient to prepare.
In this finer-grained setting, what can we say about the difference between quantum and classical computation?
Are there classical analogues of the kinds of quantum precomputation that we have proposed, or are there some types of precomputation are uniquely quantum mechanical?
Conversely, classical precomputation is widely applicable in situations where the precomputed information is used multiple times.
One could interpret recent shadow tomography proposals as examples of quantum precomputation that allow for information reuse~\cite{Aaronson2019-mk, Chen2022-ys}, and it would be interesting to see if techniques from that domain can be adapted to enable such reuse in the context of other types of quantum precomputation.

Finally, are there other, perhaps more general, classes of quantum computation that we can accelerate in the precomputation cost model?
Many proposed applications of quantum machine learning techniques to classical data rely on quantum random access memory (QRAM) to obtain a computational advantage~\cite{Biamonte2017-rm}.
Are there real-world applications where it would be natural to circumvent the need for QRAM by encoding some classical data into quantum states ahead of time?

\section*{Acknowledgements}
The authors thank Rolando Somma, Robin Kothari, Nathan Wiebe, Craig Gidney, Yigit Subasi, Ryan Babbush, 
and others
for helpful discussions and feedback.


\begin{thebibliography}{50}
\providecommand{\natexlab}[1]{#1}
\providecommand{\url}[1]{\texttt{#1}}
\expandafter\ifx\csname urlstyle\endcsname\relax
  \providecommand{\doi}[1]{doi: #1}\else
  \providecommand{\doi}{doi: \begingroup \urlstyle{rm}\Url}\fi

\bibitem[Aaronson(2004)]{Aaronson2004-xr}
S~Aaronson.
\newblock Limitations of quantum advice and one-way communication.
\newblock In \emph{Proceedings. 19th IEEE Annual Conference on Computational
  Complexity, 2004}, pages 320--332. IEEE, 2004.
\newblock ISBN 9780769521206.
\newblock \doi{10.1109/ccc.2004.1313854}.

\bibitem[Aaronson and Ambainis(2015)]{Aaronson2015-zu}
Scott Aaronson and Andris Ambainis.
\newblock Forrelation.
\newblock In \emph{Proceedings of the forty-seventh annual ACM symposium on
  Theory of Computing}, STOC '15, pages 307--316, New York, NY, USA, 14~June
  2015. ACM.
\newblock ISBN 9781450335362.
\newblock \doi{10.1145/2746539.2746547}.

\bibitem[Aaronson and Rothblum(2019)]{Aaronson2019-mk}
Scott Aaronson and Guy~N Rothblum.
\newblock Gentle measurement of quantum states and differential privacy.
\newblock 18~April 2019.
\newblock URL \url{http://arxiv.org/abs/1904.08747}.

\bibitem[Babbush et~al.(2021)Babbush, McClean, Newman, Gidney, Boixo, and
  Neven]{Babbush2021-aq}
Ryan Babbush, Jarrod~R McClean, Michael Newman, Craig Gidney, Sergio Boixo, and
  Hartmut Neven.
\newblock Focus beyond quadratic speedups for error-corrected quantum
  advantage.
\newblock \emph{PRX quantum}, 2\penalty0 (1):\penalty0 010103, 29~March 2021.
\newblock ISSN 2691-3399.
\newblock \doi{10.1103/prxquantum.2.010103}.

\bibitem[Bernstein and Lange(2013)]{Bernstein2013-mn}
Daniel~J Bernstein and Tanja Lange.
\newblock Non-uniform cracks in the concrete: The power of free precomputation.
\newblock In \emph{Advances in Cryptology - ASIACRYPT 2013}, Lecture notes in
  computer science, pages 321--340. Springer Berlin Heidelberg, Berlin,
  Heidelberg, 2013.
\newblock ISBN 9783642420443,9783642420450.
\newblock \doi{10.1007/978-3-642-42045-0\_17}.

\bibitem[Berry et~al.(2019)Berry, Gidney, Motta, McClean, and
  Babbush]{Berry2019-qo}
Dominic~W Berry, Craig Gidney, Mario Motta, Jarrod~R McClean, and Ryan Babbush.
\newblock Qubitization of arbitrary basis quantum chemistry leveraging sparsity
  and low rank factorization.
\newblock 6~February 2019.
\newblock URL \url{http://arxiv.org/abs/1902.02134}.

\bibitem[Biamonte et~al.(2017)Biamonte, Wittek, Pancotti, Rebentrost, Wiebe,
  and Lloyd]{Biamonte2017-rm}
Jacob Biamonte, Peter Wittek, Nicola Pancotti, Patrick Rebentrost, Nathan
  Wiebe, and Seth Lloyd.
\newblock Quantum machine learning.
\newblock \emph{Nature}, 549\penalty0 (7671):\penalty0 195--202, September
  2017.
\newblock ISSN 0028-0836,1476-4687.
\newblock \doi{10.1038/nature23474}.

\bibitem[Bravyi and Kitaev(2005)]{Bravyi2005-vi}
Sergey Bravyi and Alexei Kitaev.
\newblock Universal quantum computation with ideal clifford gates and noisy
  ancillas.
\newblock \emph{Phys. Rev. A}, 71\penalty0 (2):\penalty0 022316, 22~February
  2005.
\newblock ISSN 1050-2947,1094-1622.
\newblock \doi{10.1103/physreva.71.022316}.

\bibitem[Bravyi and Maslov(2021)]{Bravyi2021-px}
Sergey Bravyi and Dmitri Maslov.
\newblock Hadamard-free circuits expose the structure of the clifford group.
\newblock \emph{IEEE Trans. Inf. Theory}, 67\penalty0 (7):\penalty0 4546--4563,
  July 2021.
\newblock ISSN 0018-9448,1557-9654.
\newblock \doi{10.1109/tit.2021.3081415}.

\bibitem[Campbell and O'Gorman(2016)]{Campbell2016-ls}
Earl~T Campbell and Joe O'Gorman.
\newblock An efficient magic state approach to small angle rotations.
\newblock 14~March 2016.
\newblock URL \url{http://arxiv.org/abs/1603.04230}.

\bibitem[Chen et~al.(2022)Chen, Cotler, Huang, and Li]{Chen2022-ys}
Sitan Chen, Jordan Cotler, Hsin-Yuan Huang, and Jerry Li.
\newblock Exponential separations between learning with and without quantum
  memory.
\newblock In \emph{2021 IEEE 62nd Annual Symposium on Foundations of Computer
  Science (FOCS)}. IEEE, February 2022.
\newblock \doi{10.1109/focs52979.2021.00063}.

\bibitem[Childs et~al.(2017)Childs, Kothari, and Somma]{Childs2017-is}
Andrew~M Childs, Robin Kothari, and Rolando~D Somma.
\newblock Quantum algorithm for systems of linear equations with exponentially
  improved dependence on precision.
\newblock \emph{SIAM J. Comput.}, 46\penalty0 (6):\penalty0 1920--1950,
  1~January 2017.
\newblock ISSN 0097-5397.
\newblock \doi{10.1137/16M1087072}.

\bibitem[Cody~Jones et~al.(2012)Cody~Jones, Whitfield, McMahon, Yung,
  Van~Meter, Aspuru-Guzik, and Yamamoto]{Cody-Jones2012-tm}
N~Cody~Jones, James~D Whitfield, Peter~L McMahon, Man-Hong Yung, Rodney
  Van~Meter, Alán Aspuru-Guzik, and Yoshihisa Yamamoto.
\newblock Faster quantum chemistry simulation on fault-tolerant quantum
  computers.
\newblock \emph{New J. Phys.}, 14\penalty0 (11):\penalty0 115023, 27~November
  2012.
\newblock ISSN 1367-2630.
\newblock \doi{10.1088/1367-2630/14/11/115023}.

\bibitem[Costa et~al.(2022)Costa, An, Sanders, Su, Babbush, and
  Berry]{Costa2022-ae}
Pedro C~S Costa, Dong An, Yuval~R Sanders, Yuan Su, Ryan Babbush, and Dominic~W
  Berry.
\newblock Optimal scaling quantum linear-systems solver via discrete adiabatic
  theorem.
\newblock \emph{PRX quantum}, 3\penalty0 (4):\penalty0 040303, 7~October 2022.
\newblock ISSN 2691-3399.
\newblock \doi{10.1103/prxquantum.3.040303}.

\bibitem[Cotler et~al.(2021)Cotler, Huang, and McClean]{Cotler2021-ea}
Jordan Cotler, Hsin-Yuan Huang, and Jarrod~R McClean.
\newblock Revisiting dequantization and quantum advantage in learning tasks.
\newblock 1~December 2021.
\newblock URL \url{http://arxiv.org/abs/2112.00811}.

\bibitem[Cui et~al.(2017)Cui, Gottesman, and Krishna]{Cui2017-mu}
Shawn~X Cui, Daniel Gottesman, and Anirudh Krishna.
\newblock Diagonal gates in the clifford hierarchy.
\newblock \emph{Phys. Rev. A}, 95\penalty0 (1), 26~January 2017.
\newblock ISSN 2469-9926,2469-9934.
\newblock \doi{10.1103/physreva.95.012329}.

\bibitem[Farhi et~al.(2014)Farhi, Goldstone, and Gutmann]{Farhi2014-vk}
Edward Farhi, Jeffrey Goldstone, and Sam Gutmann.
\newblock A quantum approximate optimization algorithm.
\newblock 14~November 2014.
\newblock URL \url{http://arxiv.org/abs/1411.4028}.

\bibitem[Fowler(2012)]{Fowler2012-ti}
Austin~G Fowler.
\newblock Time-optimal quantum computation.
\newblock 17~October 2012.
\newblock URL \url{http://arxiv.org/abs/1210.4626}.

\bibitem[Gharibian and Le~Gall(2022)]{Gharibian2022-qy}
Sevag Gharibian and François Le~Gall.
\newblock Dequantizing the quantum singular value transformation: hardness and
  applications to quantum chemistry and the quantum {PCP} conjecture.
\newblock In \emph{Proceedings of the 54th Annual ACM SIGACT Symposium on
  Theory of Computing}, STOC 2022, pages 19--32, New York, NY, USA, 9~June
  2022. ACM.
\newblock ISBN 9781450392648.
\newblock \doi{10.1145/3519935.3519991}.

\bibitem[Gidney and Ekerå(2021)]{Gidney2021-ru}
Craig Gidney and Martin Ekerå.
\newblock How to factor 2048 bit {RSA} integers in 8 hours using 20 million
  noisy qubits.
\newblock \emph{Quantum}, 5\penalty0 (433):\penalty0 433, 15~April 2021.
\newblock ISSN 2521-327X.
\newblock \doi{10.22331/q-2021-04-15-433}.

\bibitem[Gidney and Fowler(2019)]{Gidney2019-qi}
Craig Gidney and Austin~G Fowler.
\newblock Flexible layout of surface code computations using {AutoCCZ} states.
\newblock 21~May 2019.
\newblock URL \url{http://arxiv.org/abs/1905.08916}.

\bibitem[Gilyén and Poremba(2022)]{Gilyen2022-gu}
András Gilyén and Alexander Poremba.
\newblock Improved quantum algorithms for fidelity estimation.
\newblock 29~March 2022.
\newblock URL \url{http://arxiv.org/abs/2203.15993}.

\bibitem[Gottesman and Chuang(1999)]{Gottesman1999-gr}
Daniel Gottesman and Isaac~L Chuang.
\newblock Quantum teleportation is a universal computational primitive.
\newblock 2~August 1999.
\newblock URL \url{http://arxiv.org/abs/quant-ph/9908010}.

\bibitem[Grady and Sinop(2008)]{Grady2008-kf}
Leo Grady and Ali~Kemal Sinop.
\newblock Fast approximate random walker segmentation using eigenvector
  precomputation.
\newblock In \emph{2008 IEEE Conference on Computer Vision and Pattern
  Recognition}, pages 1--8. IEEE, June 2008.
\newblock ISBN 9781424422425.
\newblock \doi{10.1109/cvpr.2008.4587487}.

\bibitem[Grover(1996)]{Grover1996-ah}
Lov~K Grover.
\newblock A fast quantum mechanical algorithm for database search.
\newblock In \emph{Proceedings of the twenty-eighth annual ACM symposium on
  Theory of computing - STOC '96}, STOC '96, pages 212--219, New York, New
  York, USA, 1996. ACM Press.
\newblock ISBN 9780897917858.
\newblock \doi{10.1145/237814.237866}.

\bibitem[Harrow et~al.(2009)Harrow, Hassidim, and Lloyd]{Harrow2009-iu}
Aram~W Harrow, Avinatan Hassidim, and Seth Lloyd.
\newblock Quantum algorithm for linear systems of equations.
\newblock \emph{Phys. Rev. Lett.}, 103\penalty0 (15):\penalty0 150502,
  9~October 2009.
\newblock ISSN 0031-9007,1079-7114.
\newblock \doi{10.1103/PhysRevLett.103.150502}.

\bibitem[Huang et~al.(2022)Huang, Broughton, Cotler, Chen, Li, Mohseni, Neven,
  Babbush, Kueng, Preskill, and McClean]{Huang2022-jp}
Hsin-Yuan Huang, Michael Broughton, Jordan Cotler, Sitan Chen, Jerry Li, Masoud
  Mohseni, Hartmut Neven, Ryan Babbush, Richard Kueng, John Preskill, and
  Jarrod~R McClean.
\newblock Quantum advantage in learning from experiments.
\newblock \emph{Science}, 376\penalty0 (6598):\penalty0 1182--1186, 10~June
  2022.
\newblock ISSN 0036-8075,1095-9203.
\newblock \doi{10.1126/science.abn7293}.

\bibitem[Jones(2013)]{Jones2013-cf}
Cody Jones.
\newblock Distillation protocols for fourier states in quantum computing.
\newblock 12~March 2013.
\newblock URL \url{http://arxiv.org/abs/1303.3066}.

\bibitem[Kallaugher(2022)]{Kallaugher2022-nv}
John Kallaugher.
\newblock A quantum advantage for a natural streaming problem.
\newblock In \emph{2021 IEEE 62nd Annual Symposium on Foundations of Computer
  Science (FOCS)}, pages 897--908. IEEE, February 2022.
\newblock \doi{10.1109/focs52979.2021.00091}.

\bibitem[Karp and Lipton(1980)]{Karp1980-hr}
Richard~M Karp and Richard~J Lipton.
\newblock Some connections between nonuniform and uniform complexity classes.
\newblock In \emph{Proceedings of the twelfth annual ACM symposium on Theory of
  computing - STOC '80}, STOC '80, pages 302--309, New York, New York, USA,
  28~April 1980. ACM Press.
\newblock ISBN 9780897910170.
\newblock \doi{10.1145/800141.804678}.

\bibitem[Kimmel et~al.(2017)Kimmel, Lin, Low, Ozols, and Yoder]{Kimmel2017-af}
Shelby Kimmel, Cedric Yen-Yu Lin, Guang~Hao Low, Maris Ozols, and Theodore~J
  Yoder.
\newblock Hamiltonian simulation with optimal sample complexity.
\newblock \emph{Npj Quantum Inf.}, 3\penalty0 (1):\penalty0 1--7, 30~March
  2017.
\newblock ISSN 2056-6387,2056-6387.
\newblock \doi{10.1038/s41534-017-0013-7}.

\bibitem[Le~Gall(2006)]{Le-Gall2006-nc}
François Le~Gall.
\newblock Exponential separation of quantum and classical online space
  complexity.
\newblock In \emph{Proceedings of the eighteenth annual ACM symposium on
  Parallelism in algorithms and architectures}, SPAA '06, pages 67--73, New
  York, NY, USA, 30~July 2006. ACM.
\newblock ISBN 9781595934529.
\newblock \doi{10.1145/1148109.1148119}.

\bibitem[Lin and Tong(2020)]{Lin2020-ib}
Lin Lin and Yu~Tong.
\newblock Optimal polynomial based quantum eigenstate filtering with
  application to solving quantum linear systems.
\newblock \emph{Quantum}, 4\penalty0 (361):\penalty0 361, 11~November 2020.
\newblock ISSN 2521-327X.
\newblock \doi{10.22331/q-2020-11-11-361}.

\bibitem[Litinski(2019{\natexlab{a}})]{Litinski2019-ek}
Daniel Litinski.
\newblock Magic state distillation: Not as costly as you think.
\newblock \emph{Quantum}, 3\penalty0 (205):\penalty0 205, 2~December
  2019{\natexlab{a}}.
\newblock ISSN 2521-327X.
\newblock \doi{10.22331/q-2019-12-02-205}.

\bibitem[Litinski(2019{\natexlab{b}})]{Litinski2019-nu}
Daniel Litinski.
\newblock A game of surface codes: Large-scale quantum computing with lattice
  surgery.
\newblock \emph{Quantum}, 3\penalty0 (128):\penalty0 128, 5~March
  2019{\natexlab{b}}.
\newblock ISSN 2521-327X.
\newblock \doi{10.22331/q-2019-03-05-128}.

\bibitem[Lloyd et~al.(2014)Lloyd, Mohseni, and Rebentrost]{Lloyd2014-td}
Seth Lloyd, Masoud Mohseni, and Patrick Rebentrost.
\newblock Quantum principal component analysis.
\newblock \emph{Nat. Phys.}, 10\penalty0 (9):\penalty0 631--633, 27~September
  2014.
\newblock ISSN 1745-2473,1745-2481.
\newblock \doi{10.1038/nphys3029}.

\bibitem[Martyn et~al.(2021)Martyn, Rossi, Tan, and Chuang]{Martyn2021-mf}
John~M Martyn, Zane~M Rossi, Andrew~K Tan, and Isaac~L Chuang.
\newblock Grand unification of quantum algorithms.
\newblock \emph{PRX quantum}, 2\penalty0 (4):\penalty0 040203, 3~December 2021.
\newblock ISSN 2691-3399.
\newblock \doi{10.1103/prxquantum.2.040203}.

\bibitem[Marvian and Lloyd(2016)]{Marvian2016-ya}
Iman Marvian and Seth Lloyd.
\newblock Universal quantum emulator.
\newblock 8~June 2016.
\newblock URL \url{http://arxiv.org/abs/1606.02734}.

\bibitem[Motzoi et~al.(2017)Motzoi, Kaicher, and Wilhelm]{Motzoi2017-rd}
F~Motzoi, M~P Kaicher, and F~K Wilhelm.
\newblock Linear and logarithmic time compositions of quantum many-body
  operators.
\newblock \emph{Phys. Rev. Lett.}, 119\penalty0 (16):\penalty0 160503,
  20~October 2017.
\newblock ISSN 0031-9007,1079-7114.
\newblock \doi{10.1103/PhysRevLett.119.160503}.

\bibitem[Nielsen(2004)]{Nielsen2004-oq}
Michael~A Nielsen.
\newblock Optical quantum computation using cluster states.
\newblock \emph{Phys. Rev. Lett.}, 93\penalty0 (4):\penalty0 040503, 23~July
  2004.
\newblock ISSN 0031-9007,1079-7114.
\newblock \doi{10.1103/PhysRevLett.93.040503}.

\bibitem[O'Gorman et~al.(2019)O'Gorman, Huggins, Rieffel, and
  Whaley]{O-Gorman2019-dj}
Bryan O'Gorman, William~J Huggins, Eleanor~G Rieffel, and K~Birgitta Whaley.
\newblock Generalized swap networks for near-term quantum computing.
\newblock 13~May 2019.
\newblock URL \url{http://arxiv.org/abs/1905.05118}.

\bibitem[Pham and Svore(2012)]{Pham2012-rb}
Paul Pham and Krysta~M Svore.
\newblock A {2D} nearest-neighbor quantum architecture for factoring in
  polylogarithmic depth.
\newblock 27~July 2012.
\newblock URL \url{http://arxiv.org/abs/1207.6655}.

\bibitem[Raussendorf and Briegel(2001)]{Raussendorf2001-na}
R~Raussendorf and H~J Briegel.
\newblock A one-way quantum computer.
\newblock \emph{Phys. Rev. Lett.}, 86\penalty0 (22):\penalty0 5188--5191,
  28~May 2001.
\newblock ISSN 0031-9007,1079-7114.
\newblock \doi{10.1103/PhysRevLett.86.5188}.

\bibitem[Sanders et~al.(2020)Sanders, Berry, Costa, Tessler, Wiebe, Gidney,
  Neven, and Babbush]{Sanders2020-lf}
Yuval~R Sanders, Dominic~W Berry, Pedro C~S Costa, Louis~W Tessler, Nathan
  Wiebe, Craig Gidney, Hartmut Neven, and Ryan Babbush.
\newblock Compilation of fault-tolerant quantum heuristics for combinatorial
  optimization.
\newblock \emph{PRX quantum}, 1\penalty0 (2):\penalty0 020312, 9~November 2020.
\newblock ISSN 2691-3399.
\newblock \doi{10.1103/prxquantum.1.020312}.

\bibitem[Shepherd and Bremner(2009)]{Shepherd2009-rw}
Dan Shepherd and Michael~J Bremner.
\newblock Temporally unstructured quantum computation.
\newblock \emph{Proc. Math. Phys. Eng. Sci.}, 465\penalty0 (2105):\penalty0
  1413--1439, 8~May 2009.
\newblock ISSN 1364-5021,1471-2946.
\newblock \doi{10.1098/rspa.2008.0443}.

\bibitem[Sloan et~al.(2002)Sloan, Kautz, and Snyder]{Sloan2002-zw}
Peter-Pike Sloan, Jan Kautz, and John Snyder.
\newblock Precomputed radiance transfer for real-time rendering in dynamic,
  low-frequency lighting environments.
\newblock In \emph{Proceedings of the 29th annual conference on Computer
  graphics and interactive techniques}, SIGGRAPH '02, pages 527--536, New York,
  NY, USA, 1~July 2002. ACM.
\newblock ISBN 9781581135213.
\newblock \doi{10.1145/566570.566612}.

\bibitem[Smith(1998)]{Smith1998-gk}
James~E Smith.
\newblock A study of branch prediction strategies.
\newblock In \emph{25 years of the international symposia on Computer
  architecture (selected papers)}, ISCA '98, pages 202--215, New York, NY, USA,
  1~August 1998. ACM.
\newblock ISBN 9781581130584.
\newblock \doi{10.1145/285930.285980}.

\bibitem[Somma and Subaşı(2021)]{Somma2021-ql}
Rolando~D Somma and Yiğit Subaşı.
\newblock Complexity of quantum state verification in the quantum linear
  systems problem.
\newblock \emph{PRX quantum}, 2\penalty0 (1):\penalty0 010315, 27~January 2021.
\newblock ISSN 2691-3399.
\newblock \doi{10.1103/prxquantum.2.010315}.

\bibitem[Terhal(2015)]{Terhal2015-dg}
Barbara~M Terhal.
\newblock Quantum error correction for quantum memories.
\newblock \emph{Rev. Mod. Phys.}, 87\penalty0 (2):\penalty0 307--346, 7~April
  2015.
\newblock ISSN 0034-6861,1539-0756.
\newblock \doi{10.1103/revmodphys.87.307}.

\bibitem[Zhou et~al.(2000)Zhou, Leung, and Chuang]{Zhou2000-yi}
Xinlan Zhou, Debbie~W Leung, and Isaac~L Chuang.
\newblock Methodology for quantum logic gate construction.
\newblock \emph{Phys. Rev. A}, 62\penalty0 (5), 18~October 2000.
\newblock ISSN 1050-2947,1094-1622.
\newblock \doi{10.1103/physreva.62.052316}.

\end{thebibliography}

\appendix

\section{Precomputation and quantum advice}
\label{app:quantum_advice}

The purpose of this appendix is to relate our proposed model of quantum precomputation to the notion of quantum advice and the complexity class BQP/qpoly.
We do not aim to provide a self-contained introduction to quantum complexity theory, but we will briefly mention some basic definitions that will aid in making the comparison.
The most well-studied computational problems in complexity theory are decision problems, questions that have a yes or no answer.
We can formalize a decision problem as a language, a set of bitstrings that encode the inputs to the problem for which the answer is yes.
Informally, a decision problem is in the complexity class BQP if it can be solved in polynomial time on a quantum computer.
Formally, we have the following definition:
\begin{definition}
    Let \(\left\{0, 1\right\}^*\) denote the set of all binary strings.
    A language \(L \subseteq \left\{ 0, 1 \right\}^*\) is in BQP if these exists a uniform family of polynomial-size quantum circuits, \(\left\{ C_n \right\}\), such that the following conditions hold for all \(x \in \left\{ 0, 1 \right\}^n\):
    \begin{enumerate}
        \item If \(x \in L\), then the probability that the first qubit is measured to be \(\ket{1}\) after \(C_n\) is applied to the input \(\ket{x}\otimes\ket{0\cdots0}\) is at least \(2/3\).
        \item If \(x \notin L\), then the probability that the first qubit is measured to be \(\ket{1}\) after \(C_n\) is applied to the input \(\ket{x}\otimes\ket{0\cdots0}\) is at most \(1/3\).
    \end{enumerate}
\end{definition}
Note that the circuit \(C_n\) depends only on \(n\), the size of the input. The condition that the family of circuits is uniform essentially requires that a polynomial time classical computer can generate the description of the circuit that the quantum computer will execute.

Like our model of quantum precomputation, the complexity class BQP/qpoly is intended to capture the power of a polynomial-time quantum machine augmented with an additional resource state.
Formally, the class can be defined as follows:
\begin{definition}
    A language \(L \subseteq \left\{ 0, 1 \right\}^*\) is in BQP/qpoly if there exists a uniform family of polynomial-size quantum circuits, \(\left\{ C_n \right\}\), and a family of polynomial-size quantum states, \(\left\{ \ket{\psi_n} \right\}\), such that the following conditions hold for all \(x \in \left\{ 0, 1 \right\}^n\):
    \begin{enumerate}
        \item If \(x \in L\), then the probability that the first qubit is measured to be \(\ket{1}\) after \(C_n\) is applied to the input \(\ket{x}\otimes\ket{0\cdots0}\otimes{\ket{\psi_n}}\) is at least \(2/3\).
        \item If \(x \notin L\), then the probability that the first qubit is measured to be \(\ket{1}\) after \(C_n\) is applied to the input \(\ket{x}\otimes\ket{0\cdots0}\otimes \ket{\psi_n}\) is at most \(1/3\).
    \end{enumerate}
\end{definition}
It is important to note that the additional quantum resources afforded to the polynomially powerful quantum machine can be arbitrarily complex states on \(\text{poly}(n)\) qubits. However, these states are only allowed to depend on the size of the input.

There are therefore three key differences between the model of computation considered in BQP/qpoly and the model we consider when we allow for ``free'' polynomial-time quantum precomputation.
First of all, we have defined quantum precomputation to allow inputs and outputs that are combinations of classical and quantum information. BQP/qpoly is concerned with machines that take a classical bitstring as an input and return (with some probability of failure) a single classical bit as output.
Secondly, in the precomputation cost model, we require that the quantum resources states are preparable in polynomial time, whereas the quantum advice states allowed in BQP/qpoly can be arbitrary quantum states.
Finally, in the precomputation model, we partition the input into two subsets and allow for the resource state to depend on one subset, but not the other. The complexity class BQP/qpoly only allows for the resource states to depend on the size of the input, but none of its other features.

\section{Algorithmic Primitives}

\subsection{Density matrix exponentiation}
\label{app:denmat_review}

Density matrix exponentiation is a technique that allows one to consume copies of a mixed quantum state \(\rho\) in order to approximately implement the unitary \(e^{-it\rho}\)~\cite{Lloyd2014-td}.
In \citen{Lloyd2014-td}, Lloyd et al. gave a protocol for implementing \(e^{-it\rho}\) to within an error \(\epsilon\) (in the diamond norm) by consuming
\begin{equation}
    \label{eq:denmat_exp_scaling_app}
    m = \bigo{t^2/\epsilon}
\end{equation}
copies of \(\rho\).
This scaling is optimal with respect to \(\epsilon\), and optimal with respect to \(t\) for general \(\rho\) (but not necessarily for pure states)~\cite{Kimmel2017-af}.
Furthermore, the protocol is relatively simple to implement.
In order to act on an input state \(\sigma\), one repeatedly consumes a single copy of \(\rho\) to apply an approximation to \(e^{it\rho/m}\).
This is done by performing a partial swap operator (with a small angle) on the joint system \(\rho \otimes \sigma\) and discarding the first register.
The entire evolution can be performed using \(\mathcal{O}(n t^2 /\epsilon)\) one- and two-qubit gates~\cite{Kimmel2017-af}.

Density matrix exponentiation is a basic algorithmic primitive that has been applied in a variety of ways~\cite{Lloyd2014-td, Marvian2016-ya, Gilyen2022-gu}.
In the original paper, \citen{Lloyd2014-td}, it was used as a building block in the quantum principle component analysis algorithm.
Quantum principle component analysis allows one to (approximately) sample the eigenvectors of \(\rho\) corresponding to large eigenvalues exponentially more quickly than any classical algorithm that has access only to single copies of \(\rho\)~\cite{Huang2022-jp, Cotler2021-ea}.
In \citen{Marvian2016-ya}, density matrix exponentiation was used to efficiently emulate the action of a unitary \(U\) on a small subspace by consuming samples of the form \(\ket{b} \otimes U\ket{b}\), where the input states \(\ket{b}\) span the subspace.
This type of application closely resembles a sort of quantum lookup table, and shares some features with our proposed use of density matrix exponentiation for precomputation, although the aim of that work is different.

\subsection{Gate teleportation and the Clifford hierarchy}
\label{app:review_gate_teleportation_clifford}

Our work makes heavy use of the concept of gate teleportation~\cite{Gottesman1999-gr}.
We illustrated the single-qubit version of gate teleportation in \Cref{fig:generic_teleportation} in the main text, but we present a more detailed review here.
Given a unitary \(U\), gate teleportation allows us to prepare a resource state \(\Gamma(U)\) that we can later consume to apply \(U P\) to an arbitrary state \(\ket{\psi}\), where the ``byproduct operator'' \(P\) is an element of the Pauli group randomly determined by the measurement outcomes of the teleportation protocol.
The state obtained when using gate teleportation to apply \(U\) (actually \(UP\)) to \(\ket{\psi}\) can be written as \(\left( U P U^\dagger\right) U \ket{\psi}\). Multiplying by \(U P^\dagger U^\dagger\) yields \(U \ket{\psi}\).

Gate teleportation can be especially useful when \(U P^\dagger U^\dagger\) is simpler to apply than \(U\) itself.
This is the case in the canonical application of gate teleportation, implementing \(T\) gates in a quantum error correcting code that supports fault-tolerant Clifford gates~\cite{Bravyi2005-vi}.
The problem of applying \(T\) gates without error is reduced to the problem of preparing high-fidelity ``magic states,'' because, for all possible byproduct operators \(P\), \(T P T^\dagger\) is a Clifford gate despite the fact that \(T\) is not.\footnote{
In practice, \(T\) gates can actually be implemented using a simpler and more specialized form of gate teleportation known as one-bit teleportation~\cite{Zhou2000-yi}, but for our purposes we can ignore this detail.}
Just as state teleportation trivially generalizes to multiple qubits, gate teleportation can likewise be straightforwardly applied to multiple qubits.
In the \(n\)-qubit case, the byproduct operator is an \(n\)-qubit Pauli operator (up to a phase) that depends on the \(2n\)-bit measurement outcome obtained from \(n\) simultaneous bell basis measurements.

The notion that gate teleportation is most useful when \(U P^\dagger U^\dagger\) is easier to implement than \(U\) itself led Gottesman and Chuang to define an infinite hierarchy of unitaries now known as the Clifford hierarchy~\cite{Gottesman1999-gr}.
The first level of the Clifford hierarchy, which we denote by \(\mathcal{C}^{(1)}\), is defined to be the Pauli group.
The \(k\)th level of the Clifford hierarchy is defined inductively,
\begin{equation}
    \label{eq:clifford_hierarchy_induction}
    \mathcal{C}^{(k)} \coloneqq \left\{ U | U P U^\dagger \in \mathcal{C}^{(k-1)} \; \forall P \in \mathcal{C}^{(1)} \right\}.
\end{equation}
The second level of the hierarchy is therefore the usual Clifford group.
The higher levels of the Clifford hierarchy are harder to characterize in familiar terms, but we can give some examples.
For instance, \(T\) gates, Toffoli gates, and \(CCZ\) gates belong to \(\mathcal{C}^{(3)}\).
More generally, multi-controlled \(C^{k-1}NOT\) and \(C^{k-1}Z\) gates are in \(\mathcal{C}^{(k)}\), as are the single-qubit rotations \(Z_k\), 
\begin{equation}
    Z_k \defeq
    \begin{bmatrix}
        1 & 0
        \\
        0 & e^{i\pi 2^{-k + 1}}
    \end{bmatrix}.
\end{equation}
It is an open problem to fully characterize the higher levels of the hierarchy, although the diagonal elements are well-understood algebraically in terms of polynomials and roots of unity~\cite{Cui2017-mu}.

\subsection{Review of selective teleportation}
\label{app:selective_teleportation}

When gate teleportation is used to implement a unitary \(U\) that is not in the Clifford group, the resulting correction operator \(U P^\dagger U^\dagger\) is not, in general, a Pauli operator.
For example, consider the use of gate teleportation to implement a \(T\) gate.
With probability \(\frac{1}{2}\), correcting for the byproduct operator requires the subsequent implementation of a phase gate (\(S \in \mathcal{C}^{(2)}\)).
Naively, this means that after applying a \(T\) gate using gate teleportation it is necessary to determine and apply the correction before performing additional Clifford gates.
However, in \citen{Fowler2012-ti}, Fowler showed how a generalization of quantum teleportation can be used to selectively implement this phase gate correction using a small number of ancilla qubits measured in a classically controlled choice of the \(X\) or \(Z\) basis.

\begin{figure}[t]
    \centering
    \begin{subfigure}{.53\textwidth}
        \begin{quantikz}
            [font=\footnotesize, column sep = 3.4mm, execute at end picture={
            \node (A)[fit=(\tikzcdmatrixname-1-7)(\tikzcdmatrixname-1-8)(\tikzcdmatrixname-2-7)(\tikzcdmatrixname-2-8), inner sep=4pt, rounded corners] {};
            \draw [->, dotted, thick,-{Latex[round]}] (A.south) -- ++(0,-.4);
            \node (B)[fit=(\tikzcdmatrixname-1-9)(\tikzcdmatrixname-1-10)(\tikzcdmatrixname-2-9)(\tikzcdmatrixname-2-10), inner sep=4pt, rounded corners] {};
            \draw [->, dotted, thick,-{Latex[round]}] (B.south) -- ++(0,-1.04);
            }]
            \lstick{$\ket{\psi}$}                                                & \ctrl{1}   & \targ{}                            & \qw       & \qw \arrow[rr, thick,-{Latex[round]}, dotted, shorten >=7pt] &     & \gategroup[wires=2,steps=2,style={
            dashed,rounded corners,fill=blue!20, inner sep=0pt},background]{{A}} & \meterD{Z} & \gategroup[wires=2,steps=2,style={
            dashed,rounded corners,fill=red!20, inner sep=0pt},background]{{B}}  & \meterD{X}
            \\
            \lstick{$\ket{0}$}                                                   & \targ{}    & \qw                                & \targ{}   & \qw \arrow[rr, thick,-{Latex[round]}, dotted, shorten >=7pt] &     &                                    & \meterD{X} &     & \meterD{Z}
            \\ [2pt]
            \lstick{$\ket{+}$}                                                   & \qw        & \ctrl{-2}                          & \qw       & \qw                                                          & \qw & \qw                                & \qw        & \qw & \qw \rstick{\( P \left( {\color{Blue3}\ket{\psi}} \big/ {\color{Red3} \ket{+}} \right)\)}
            \\
            \lstick{$\ket{+}$}                                                   & \qw        & \qw                                & \ctrl{-2} & \qw                                                          & \qw & \qw                                & \qw        & \qw & \qw \rstick{\( P \left( {\color{Blue3}\ket{+}} \big/ {\color{Red3} \ket{\psi}} \right)\)}
        \end{quantikz}
        \caption{Selective destination teleportation}
        \label{fig:selective_destination_teleportation}
    \end{subfigure}
    \begin{subfigure}{.46\textwidth}
        \begin{quantikz}
            [font=\footnotesize, column sep = 3.4mm, execute at end picture={
            \node (A)[fit=(\tikzcdmatrixname-1-6)(\tikzcdmatrixname-1-7)(\tikzcdmatrixname-2-6)(\tikzcdmatrixname-2-7), inner sep=4pt, rounded corners] {};
            \draw [->, dotted, thick,-{Latex[round]}] (A.south) -- ++(0,-.48);
            \node (B)[fit=(\tikzcdmatrixname-1-8)(\tikzcdmatrixname-1-9)(\tikzcdmatrixname-2-8)(\tikzcdmatrixname-2-9), inner sep=4pt, rounded corners] {};
            \draw [->, dotted, thick,-{Latex[round]}] (B.south) -- ++(0,-.48);
            }]
            \lstick{$\ket{\alpha}$}                                              & \ctrl{2}   & \qw                                & \qw \arrow[rr, thick,-{Latex[round]}, dotted, shorten >=7pt] &     & \gategroup[wires=2,steps=2,style={
            dashed,rounded corners,fill=blue!20, inner sep=0pt},background]{{A}} & \meterD{X} & \gategroup[wires=2,steps=2,style={
            dashed,rounded corners,fill=red!20, inner sep=0pt},background]{{B}}  & \meterD{Z}
            \\
            \lstick{$\ket{\beta}$}                                               & \qw        & \ctrl{1}                           & \qw\arrow[rr, thick,-{Latex[round]}, dotted, shorten >=7pt]  &     &                                    & \meterD{Z} &     & \meterD{X}
            \\ [2pt]
            \lstick{$\ket{0}$}                                                   & \targ{}    & \targ{}                            & \qw                                                          & \qw & \qw                                & \qw        & \qw & \qw \rstick{\( P \left( {\color{Blue3}\ket{\alpha}} \big/ {\color{Red3} \ket{\beta}} \right)\)}
        \end{quantikz}
        \caption{Selective source teleportation}
        \label{fig:selective_source_teleportation}
    \end{subfigure}
    \caption{Circuit diagrams for the one-qubit versions of selective destination and source teleportation~\cite{Fowler2012-ti}. Both protocols allow for a choice that is made by selecting between two measurement settings (indicated by the blue and red shaded areas of the diagrams). Selective destination teleportation teleports the state of one qubit to a choice of two different qubits. Selective source teleportation allows one to choose which of two qubits will have its state teleported to a fixed target.
    The possible states of the output qubit(s) are color-coded to match the measurement settings that select for them.
    In both cases, a byproduct operator \(P\) drawn from the set \(\left\{ \mathbb{I}, X, Z, XZ \right\}\) is randomly applied based on the measurement outcomes.}
    \label{fig:selective_teleportation}
\end{figure}
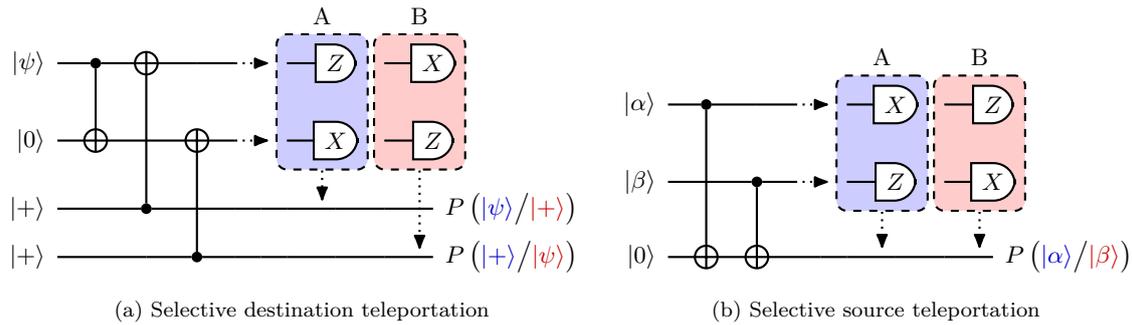

Fowler's selective teleportation relies on two related constructions, selective source teleportation and selective destination teleportation.
Selective destination teleportation allows one to teleport a single qubit's state to either one of two destination qubits.
Selective source teleportation allows for teleportation from a choice of two different source qubits to a fixed destination qubit.
Both types of selective teleportation are controlled by making an appropriate choice of measurement basis and both introduce a Pauli byproduct operator \(P \in \left\{ \mathbb{I}, X, Z, XZ \right\}\) that can be inferred from the (uniformly random) measurement outcomes.
We give circuit diagrams for the single-qubit versions of these primitives in \Cref{fig:selective_teleportation}. The multi-qubit versions are straightforward generalizations.

Together, selective source and destination teleportation can be used to implement a primitive that we refer to as selective gate teleportation.
We illustrated the single-qubit version of this selective gate teleportation in \Cref{fig:selective_gate_teleportation} in the main text.
Selective gate teleportation allows us to apply our choice of unitaries \(U_1\) or \(U_2\) to an unknown \(n\)-qubit state \(\ket{\psi}\) by choosing how to measure some set of \(4n\) ancilla qubits. As a special case, we can use selective gate teleportation to defer the choice of whether or not to apply a unitary \(U\) by taking \(U_1 = U\) and \(U_2 = \mathbb{I}\).
Selective gate teleportation randomly introduces the byproduct operators \(P^{(1)}\) and \(P^{(2)}\) (both \(n\)-qubit Pauli operators) before and after the location at which the choice of unitaries is to be applied.
For example, let \(s \in \left\{ 0, 1 \right\}\) denote the classical bit that determines whether or not to perform the teleportation that applies \(U\).
Rather than obtaining the desired \(U^s \ket{\psi}\), we instead obtain the state
    \(\ket{\phi} = P^{(1)} U^s P^{(2)} \ket{\psi}\).
To obtain \(U^s \ket{\psi}\), we would need to subsequently apply the correction operator \(U^{s} P^{(2)^\dagger} U^{s\dagger} P^{(1)^\dagger} \).

In the case that Fowler originally consider in \citen{Fowler2012-ti}, one first uses gate teleportation to implement a \(T\) gate (up to a possible \(S\) gate correction) and then selectively applies the \(S\) gate.
Because \(S\) is a Clifford gate, \(S^{s} P^{(2)^\dagger} S^{s\dagger} P^{(1)^\dagger} \) is a Pauli operator regardless of the choice of \(s\) or the measurement outcomes.
As a consequence, the measurements for both teleportation steps can be deferred or performed while applying additional Clifford gates and the necessary Pauli correction can be propagated through the resulting circuit afterwards.
This type of optimization has been used to create efficient surface code layouts for a variety of algorithmic primitives~\cite{Fowler2012-ti, Litinski2019-nu, Gidney2019-qi}.

\section{\texorpdfstring{The \(\mathcal{Z}^{(k)}\) hierarchy}{The controlled-Z hierarchy}}
\label{app:Zk_hierarchy}

In \Cref{sec:selective_gate_teleportation_precomputation}, we defined \(\mathcal{Z}^{(k)}\) to be the set of \(n\)-qubit unitaries generated by arbitrary products of controlled \(Z\) gates with up to \(k-1\) control qubits (including the case with \(0\) controls, \(Z\) gates themselves) and \(\pm \mathbb{I}\).
For convenience, we define \(\mathcal{Z}^{(0)} \coloneqq \left\{ \pm \mathbb{I} \right\}\).
Let \(\mathcal{D}^{(k)}\) denote the elements of the \(k\)-th level of the Clifford hierarchy that are also diagonal.
As sets, we have that \(\mathcal{Z}^{(k)} \subseteq \mathcal{D}^{(k)} \subset \mathcal{C}^{(k)}\).
While \(\mathcal{C}^{(k)}\) does not form a group for \(k > 2\), \citen{Cui2017-mu} showed that \(\mathcal{D}^{(k)}\) is a group for all \(k\).

The set \(\mathcal{Z}^{(k)}\) can also be shown to form a group under composition.
By definition, \(\mathcal{Z}^{(k)}\) is closed under composition (which is associative) and includes the identity element.
Because diagonal unitaries commute and \(C^{k}Z\) gates are self-inverse for all \(k\), we can see that each element of \(\mathcal{Z}^{(k)}\) is its own inverse.
Therefore, \(\mathcal{Z}^{(k)}\) is a group.

The following proposition will be useful:
\GXcommutation*
\begin{proof}
    We will prove this proposition by induction.
    The \(k=0\) case is clear by inspection and the \(k=1\) case follows from the fact that Pauli operators either commute or anti-commute.
    Now let us assume that the proposition is true for all \(j < k\) and prove that it must also hold for \(j = k\).
    Consider an arbitrary \(G \in \mathcal{Z}^{(k)}\) and \(\boldsymbol{s} \in [n]\).

    First of all, we can simplify the proof by considering a single Pauli \(X\) operator acting on arbitrary qubit \(i\) rather than the product \(X_{\boldsymbol{s}}\).
    This is because we can expand \(X_{\boldsymbol{s}} G X_{\boldsymbol{s}} G^\dagger\) as \(X_{\boldsymbol{s}} X_{\boldsymbol{s}_1} \; X_{\boldsymbol{s}_1} G X_{\boldsymbol{s}_1} G^\dagger \; X_{\boldsymbol{s}_2} \; X_{\boldsymbol{s}_2} G X_{\boldsymbol{s}_2} G^\dagger \cdots G^\dagger\) through repeated resolutions of the identity.
    If we can show that \(X_i G X_i G^\dagger \in \mathcal{Z}^{(k-1)}\) for all \(i\), then it would follow that
    \begin{equation}
        X_{\boldsymbol{s}} G X_{\boldsymbol{s}} G^\dagger =  X_{\boldsymbol{s}} X_{\boldsymbol{s}_1} G'_1 X_{\boldsymbol{s}_2} G'_2 \cdots
    \end{equation}
    for some set of \(\left\{ G'_1, G'_2, \cdots \right\} \subseteq \mathcal{Z}^{(k-1)}\).
    We could then use the inductive hypothesis to commute the various \(X\) operators through to the left, incurring additional terms from the \(\mathcal{Z}^{(j)}\) hierarchy with \(j < k\).
    These are all elements of \(\mathcal{Z}^{(k-1)}\), which is a group, and therefore their product is also in \(\mathcal{Z}^{(k-1)}\).
    The \(X\) terms would cancel, completing the proof.

    With that simplification established, the task that remains is to show that 
    \begin{equation}
        \label{eq:single_qubit_xg}
        X_i G X_i G^\dagger \in \mathcal{Z}^{(k-1)}
    \end{equation} for an arbitrary qubit \(i\).
    We can further simplify by expanding \(G\) as a product of \(m\) unitaries that are either \(\pm \idmat\), single-qubit \(Z\) gates, or \(C^jZ\) gates (for \(j < k\)),
    \begin{equation}
        G = \prod_{\ell=1}^m G_\ell.
    \end{equation}
    We will proceed by showing that \(G_\ell X_i = X_i G'_\ell G_\ell\) for some \(G'_\ell \in \mathcal{Z}^{(k-1)}\).
    If this statement holds, then we can commute \(X_i\) to the left through the each of the \(G_\ell\) terms that make up \(G\) in \Cref{eq:single_qubit_xg} and cancel it, picking up a collection of additional \(G'_\ell\) terms from \(\mathcal{Z}^{(k-1)}\).
    Because diagonal unitaries commute, we could also commute these additional terms to the left through the \(G_\ell\) terms, allowing \(G\) and \(G^\dagger\) to cancel and leaving us with a product of \(G'_\ell\) terms.
    Because \(\mathcal{Z}^{(k-1)}\) is a group, this product of \(G'_\ell\) terms would be in \(\mathcal{Z}^{(k-1)}\) and we would therefore be done.

    Now all that remains is to show that
    \begin{equation}
        \label{eq:last_statement}
        G_\ell X_i = X_i G'_\ell G_\ell
    \end{equation}    
    for some \(G'_\ell \in \mathcal{Z}^{(k-1)}\).
    First consider the case where \(G_\ell\) and \(X_i\) have support on disjoint qubits.
    Then we trivially have \(G_\ell X_i = X_i G_\ell\), which shows that the equality in \Cref{eq:last_statement} holds if we take \(G'_\ell = \idmat\).
    Now we address the case where \(X_i\) acts on one of the qubits that \(G_\ell\) also acts non-trivially on.
    Let \(\boldsymbol{x}\) denote the indices of the qubits where \(G_\ell\) acts non-trivially.

    Consider the action of the operator \(X_i G_\ell X_i G_\ell\) on an arbitrary state \(\ket{\psi}\).
    Applying \(G_\ell\) flips the sign of those computational basis states where the qubits index by \(\boldsymbol{x}\) are all in the \(1\) state.
    Applying \(X_i\) flips the state of the \(i\)th qubit.
    Applying \(G_\ell\) once again flips the sign of those basis states where the qubits index by \(\boldsymbol{x}\) are all in the \(1\) state.
    Applying \(X_i\) unflips the state of the \(i\)th qubit.
    The cumulative result of these operations is to flip the sign of those states index by the qubits in the set \(\boldsymbol{x}\setminus i\).
    In other words, \(X_i G_\ell X_i G_\ell\) acts as a controlled \(Z\) operator with one fewer controls than \(G_\ell\) (the control on qubit \(i\) is removed).
    Letting \(G'_\ell\) denote this new operator, we have that \(G'_\ell \in \mathcal{Z}^{(k-1)}\) by the definition of \(\mathcal{Z}^{(k-1)}\).
    We can multiply the expression \(G'_\ell = X_i G_\ell X_i G_\ell\) by \(X_i\) on the left and \(G_\ell\) on the right to obtain the desired result,
    \begin{equation}
        X_i G'_\ell G_\ell = G_\ell X_i.
    \end{equation}

    This completes the proof.
\end{proof}

\end{document}